\documentclass[runningheads]{llncs}
\usepackage[T1]{fontenc}
\usepackage{graphicx}
\usepackage[outline]{contour}
\usepackage[skins]{tcolorbox}
\usepackage{quiver}
\usepackage{tikz}
\usepackage{tikz-cd}
\usepackage{hyperref}

\usepackage{color}

\usepackage{makecell}
\usepackage{amsmath}
\usepackage{amssymb}
\usepackage{amsfonts}

\usepackage{algorithm}
\usepackage{algpseudocode}
\algnewcommand\Input{\textbf{Input: }}
\algnewcommand\Output{\textbf{Output: }}
\algdef{SE}[MATCH]{Match}{EndMatch}[1]{\textbf{match} #1}{\textbf{end match}}
\algdef{SE}[CASE]{Case}{EndCase}[1]{\textbf{case #1:} }{}
\algtext*{EndCase}

\algtext*{EndFunction}

\renewcommand{\phi} {\varphi}

\newcommand{\AP}[0]{\mathit{AP}}
\newcommand{\NP}[0]{\mathit{NP}}
\newcommand{\defeq}[0]{\overset{\mathrm{def}}{=}}

\newcommand{\opX}[0]{\mathop{\bigcirc}}
\newcommand{\opWX}[0]{\mathop{\text{\textcircled{\scriptsize $\mathsf{w}$}}}}
\newcommand{\opU}[0]{\mathbin{\mathcal{U}}}
\newcommand{\opBox}[0]{\mathop{\Box}}
\newcommand{\opDiamond}[0]{\mathop{\lozenge}}

\newcommand{\opF}[0]{\mathop{\mathsf{Front}}}
\newcommand{\opB}[0]{\mathop{\mathsf{Back}}}
\newcommand{\opL}[0]{\mathop{\mathsf{Left}}}
\newcommand{\opR}[0]{\mathop{\mathsf{Right}}}

\newcommand{\opAt}[1]{@_{#1}}
\newcommand{\opBind}[1]{\mathop{\downarrow\! #1}}

\newcommand{\Pos}[0]{\mathit{Pos}}

\newcommand{\subst}[3]{#1[#2 \mapsto #3]}

\newcommand{\safe}[0]{\mathit{safe}}

\newcommand{\POV}[0]{\mathit{POV}}
\newcommand{\SV}[0]{\mathit{SV}}

\newcommand{\A}[1]{A_{\mathit{#1}}}

\spnewtheorem{preconditions}{Preconditions}{\itshape}{}

\newif\ifcommentson\commentsonfalse
\newif\ifcommentson\commentsontrue
\ifcommentson
\newcommand{\Radu}[1]{\marginpar{\footnotesize \color{red} {\bf Radu:} \textsf{\scriptsize #1}}}
\newcommand{\Yoav}[1]{\marginpar{\footnotesize \color{red} {\bf Yoav:} \textsf{\scriptsize #1}}}
\newcommand{\Kevin}[1]{\marginpar{\footnotesize \color{red} {\bf Kevin:} \textsf{\scriptsize #1}}}
\newcommand{\Rose}[1]{\marginpar{\footnotesize \color{red} {\bf Rose:} \textsf{\scriptsize #1}}}
\newcommand{\Yusuke}[1]{\marginpar{\footnotesize \color{red} {\bf Yusuke:} \textsf{\scriptsize #1}}}
\newcommand{\Ichiro}[1]{\marginpar{\footnotesize \color{red} {\bf Ichiro:} \textsf{\scriptsize #1}}}
\else
\newcommand{\Radu}[1]{}
\newcommand{\Yoav}[1]{}
\newcommand{\Kevin}[1]{}
\newcommand{\Rose}[1]{}
\newcommand{\Yusuke}[1]{}
\newcommand{\Ichiro}[1]{}
\fi

\newcommand{\commentOut}[1]{}

\newcommand{\pagelimitmarker}[1]{~\\ {\ifthenelse{\thepage>#1}{\textcolor{red}{\Huge Exceeding the page limit}}{\marginpar{\textcolor{red}{\bf Within the page limit}}}}~\\ {\marginpar{\textcolor{red}{~Page Limit: {\bf#1}}}}~\\ {\marginpar{\textcolor{red}{~Current: $\bf\thepage$}}}}

\newif\ifcommentson\commentsonfalse
\usepackage{mathtools}

\begin{document}
\title{Hybrid Spatiotemporal Logic for Automotive Applications: Modeling and Model-Checking}
\titlerunning{Hybrid Spatiotemporal Logic for Automotive Applications}
\author{Radu-Florin Tulcan\inst{1}\orcidID{0000-0002-6966-1136} \and
Rose Bohrer\inst{2}\orcidID{0000-0001-5201-9895} \and
Yo\`av Montacute\inst{3}\orcidID{0000-0001-9814-7323}\and
Kevin Zhou \inst{3}\orcidID{0000-0003-4907-9453}  \and \\
Yusuke Kawamoto\inst{2}\orcidID{0000-0002-2151-9560} \and
Ichiro Hasuo\inst{3}\orcidID{0000-0002-8300-4650}}
\authorrunning{R. Tulcan et al.}
\institute{TU Wien, Theory and Logic Group \and
National Institute of Advanced Industrial Science and Technology (AIST) \and
National Institute of Informatics (NII)}
\maketitle
\begin{abstract}
We introduce a hybrid spatiotemporal logic for automotive safety applications (HSTL), focused on highway driving.
Spatiotemporal logic features specifications about vehicles throughout space and time, while hybrid logic enables precise references to individual vehicles and their historical positions.
We define the semantics of HSTL and provide a baseline model-checking algorithm for it.
We propose two optimized model-checking algorithms, which reduce the search space based on the reachable states and possible transitions from one state to another.
All three model-checking algorithms are evaluated on a series of common driving scenarios such as safe following, safe crossings, overtaking, and platooning.
An exponential performance improvement is observed for the optimized algorithms.
\end{abstract}


\section{Introduction}
\label{sub:intro}
Modern automotives are safety-critical cyber-physical systems (CPSs). From consumer cars to autonomous vehicles, software is a core component, responsible for protecting the safety of human passengers.
For this reason, the formal verification of these systems' correctness has evolved into a mature field with diverse approaches such as monitoring~\cite{pek2020using}, test generation~\cite{DBLP:conf/ivs/TuncaliFIK18}, model-checking~\cite{fan2016automatic}, and theorem-proving~\cite{platzer2018logical,wang2015improved,hasuo2022goal}.
This diversity of approaches represents several fundamental tradeoffs: highly rigorous approaches are often less flexible, while highly expressive approaches are often require extensive user expertise.
More subtly, different approaches are optimized for different components of an automotive software system.
High-precision models such as hybrid-dynamical systems\footnote{Hybrid-dynamical systems, which combine discrete and continuous system dynamics, should not be confused with hybrid logic, which introduces names to modal logic.} are tuned to the precise low-level decision-making needed to safely control a vehicle's physical actuators. However, this precision can become a burden for high-level motion planning, which is more combinatorial in nature and values efficient search among possible paths.
As automotive CPS increasingly enter real highways, it is essential to identify formal methods that can balance these tradeoffs, both to ensure that rigorous models can transfer into practice and to guarantee comprehensive safety from planning to control.
Discrete model-checking approaches hold promise to achieve this balance: they combine a firm foundation in formal logic with high automation.

We introduce \emph{Hybrid Spatiotemporal Logic for Automotive Safety} (\emph{HSTL}), a new logic designed to balance rigor and expressiveness with ease of model-checking.
Though the logic is quite general, its current emphasis is on highway driving.
In planning for highway driving, a discrete grid-based model of space is natural: planning is more concerned with navigating between discrete lanes and comprehending finite relationships between vehicles (such as one vehicle being in front of another).
Discrete grids enable efficient search among the many combinations of relationships between vehicles, long before precise coordinates in continuous space are ever considered.

As the name suggests, HSTL combines temporal, spatial, and hybrid logic features.
Temporal operators are widespread in CPS formal methods because they allow us to specify safety properties that hold always, goals which are achieved eventually, and combinations thereof.
Spatial operators allow specifying properties that hold at neighboring positions on a grid.
Hybrid logic operators allow assigning and using names for positions on the grid.
The resulting combination is greater than the sum of its parts, allowing complex specifications of vehicle motion over time, while maintaining a discrete model of motion that is amenable to search.
We advocate that this makes HSTL an ideal interface between planning and control: HSTL model-checking effectively produces high-level motion plans,
which can be consumed by following waypoint-based planning and eventually control stages.

We support this position by developing three increasingly sophisticated model-checking algorithms for HSTL and evaluating them on driving case studies such as following, intersection-crossing, overtaking, and platooning.
In modeling the case studies, we show HSTL is flexible enough to model a diverse range of concrete scenarios.
In evaluating algorithm performance, we show that fusing hybrid logic, spatial, and temporal operators does not increase the complexity of model-checking; on the contrary, it creates critical opportunities for optimization.
When hybrid logic operators are used to precisely specify trajectories, the model-checker can identify these trajectories and commit to them, drastically reducing its search space and improving its scalability.
We propose that the resulting algorithms are suitable for typical offline applications such as test case generation.

\section{Related Work}
\label{sub:related}

We discuss related works in environment modeling, logic, and CPS verification.

\paragraph{Environment modeling.}
The predominant model of driving environments is the scene graph (SG), which models objects' semantic relationships, including high-level spatial relationships \cite{woodlief2025closing,chang2021comprehensive} (SGs are not limited to driving applications).
Machine learning plays a significant role  both in generating SGs from raw visual data~\cite{malawade2022roadscene2vec} and in processing existing SGs, e.g., to assess risk \cite{yu2021scene} or predict behavior \cite{mylavarapu2020towards}.
Tool support for SGs is growing: the driving simulator CARLA~\cite{dosovitskiy2017carla} recently gained the SG generator CARLASGG \cite{woodlief2025closing} and the language Scenic~\cite{fremont2019scenic} supports end-user SG and scenario generation. 
More closely related to this work, temporal correctness specifications on SGs have begun to receive attention as well~\cite{toledo2024specifying}. We focus on enriching specifications to combine temporal, spatial, and hybrid reasoning, but conversely, we do not attempt to support the full SG data structures of modern SG tools like CARLASGG or Scenic.
In reasoning about scene graphs, manipulating discrete spatial relationships between objects is an important challenge; the \emph{grid-graph} structure in this paper is essentially a simplification of SGs that focuses on cardinal-direction spatial relationships amongst vehicles, lanes, and objects.

\paragraph{Hybrid and Spatiotemporal Logics.}
Hybrid logics extend modal logics with nominal formulas. In the simplest case, nominals give names to possible worlds.
The origins of hybrid logic are generally attributed to Prior's hybrid tense logic~\cite{prior1967past} and the name attributed to Blackburn's work~\cite{blackburn2000representation}, though our work is inspired more directly by Bra{\"u}ner's presentation~\cite{brauner2010hybrid}.
Though it is distinct from hybrid-dynamical systems, their combination has been explored, both to provide proof goal management~\cite{platzer2007towards} and to study information-flow security~\cite{bohrer2018hybrid}.
Our main deviation from traditional hybrid logic is that nominals do not name states, but positions.
This approach is inspired by and refines that of Multi-Lane Spatial Logic~\cite{hilscher2011abstract} and its extensions~\cite{schwammberger2018abstract,schwammberger2021extending}, a hybrid spatial logic with more limited temporal features.

Spatiotemporal logics (e.g.,~\cite{FDM2023,FDM2024,bennett2002multi,kontchakov2007spatial,cheng2021dynamic}) combine both spatial and temporal reasoning, each of which come in many varieties, leading to many potential combinations.
Topological approaches to space in applications, such as RCC8 reasoning~\cite{randell1992spatial}, are common but computationally complex~\cite{bennett2002multi}.
Advanced temporal logics like Signal Temporal Logic have been extended with spatial reasoning to provide practical formal methods like falsification \cite{li2020stsl,li2021runtime}.
We take inspiration from spatiotemporal logic on closure spaces~\cite{ciancia2015experimental,ciancia2017model} which, like us, explores practical model-checking algorithms, but differs in its operators because we are interested in grid-based vehicle motion, while it is interested in connectivity reasoning.

\paragraph{Verification of CPS.}
Verification of CPS is pursued both with logic and automata models such as hybrid automata~\cite{alur1991hybrid,fan2016automatic}.
In the automotive domain, automata have been successfully applied to runtime verification through model-checking~\cite{althoff2014online,pek2020using}.
For planning, temporal logics are widely-used~\cite{kress2009temporal,kress2018synthesis,smith2011optimal,plaku2015motion}.
The uncertainty of autonomous vehicles makes robustness~\cite{yu2022formally,schwarting2018planning,sun2019formal,lygeros2002verified,amini2020learning,wang2025robustness} an important topic for control especially but planning as well; discretization often induces conservative over-approximation, which indirectly promotes robustness.
Verification of control algorithms has seen success through program logics for hybrid systems, such as differential dynamic logic~\cite{platzer2018logical} and Hybrid Hoare Logic\cite{wang2015improved}.
The more recent \emph{differential Floyd-Hoare logic}~\cite{hasuo2022goal,eberhart2023formal} is a descendant of differential dynamic logic optimized for automotive safety in the \emph{responsibility-sensitive safety framework}~\cite{shalev2017formal}.
These program logics could potentially be combined with HSTL to verify that controllers safely implement the resulting plans.

\section{Hybrid Spatiotemporal Logic for Automotive Safety}

We introduce the syntax and semantics of \emph{Hybrid Spatiotemporal Logic for Automotive Safety} (\emph{HSTL}) and present non-trivial validities and non-validities.

\subsection{Syntax}
\label{sec:syntax}
The language of HSTL contains temporal operations (which control the flow of time), spatial operations (which locally move the viewpoint in space), and hybrid operations (which allow us to name and jump between different points in space).
\begin{definition}[Syntax]\rm
	Given a set $\AP$ of atomic propositions and a set $\NP$ of nominals, the \emph{formulas} of the logic are defined as follows:
	\begin{align*}
    \varphi ::= \top \mid &  \; \mathit{prop} \mid v \mid \neg \varphi \mid \varphi_1\land \varphi_2 \mid \opX \varphi \mid \varphi_1 \opU \varphi_2 \mid \\& \opF \varphi \mid \opB \varphi \mid \opL \varphi \mid \opR \varphi \mid \opAt{v} \varphi \mid \opBind{v} \varphi
	\end{align*}
	with $\mathit{prop} \in \AP$ atomic proposition and $v \in \NP$ a nominal.
\end{definition}
The \emph{temporal operators} $\opX$ and $\opU$ for `next' and `until', as usual, representing moving along a (finite) trace.
The \emph{spatial modalities} $\opF$, $\opB$, $\opL$, $\opR$ represent spatial movement in our representation of physical space (a grid-graph, which is defined in Section~\ref{sec:semantics}). 
The \emph{hybrid operators} manipulate nominals $v$.
In contrast to traditional hybrid logics (cf.~\cite{brauner2010hybrid}), our nominals $v$ name the locations of vehicles within grid-graph $G$, rather than possible worlds.
This leads to different design decisions: for example, the semantics of nominals $v$ vary as a function of time.
We use the binder $\opBind{v}$ to store the current position in a fresh $v$ for future use.
By convention, the nominal $\SV$ means \emph{subject vehicle} (``our vehicle'') and the nominal $\POV$ means some \emph{point-of-view vehicle} (``other vehicle'').

\subsection{Semantics}
\label{sec:semantics}
The semantics of HSTL give a formal specification of the meaning of each formula, which model-checking algorithms must faithfully implement.
The HSTL semantics extend traditional trace-based semantics for temporal logic with an argument that tracks our \emph{viewpoint} within some mathematical structure (cf.\ tracking the current world in traditional hybrid logic).
Changing the viewpoint does not move any vehicle.
The mathematical structures in which movement occurs are called \emph{grid-graphs}.
These are directed graphs in which each vertex corresponds to a cell of a grid, and each vertex has a directed edge to its four horizontal and vertical neighbors.

\begin{definition}[Grid-graph]\rm
A \emph{grid-graph} is a pair $(Pos, E)$ where $\Pos = \{p_{i,j} \mid 1 \leq i \leq m, 1 \leq j \leq n\}$ is a set of elements indexed by coordinates on an $m \times n$ grid, and $E$ is a binary relation on $\Pos$ such that, for all $1 \le i \le m$ and $1 \le j \le n$, the pairs $(p_{i,j}, p_{i,j+1})$, $(p_{i,j}, p_{i,j-1})$, $(p_{i,j}, p_{i+1,j})$, and $(p_{i,j}, p_{i-1,j})$, whenever defined, are elements of $E$.
\vspace{-1em}
\begin{figure}
    \centering
    \begin{tikzcd}
	\bullet & \bullet & \bullet & \bullet \\
	\bullet & \bullet & \bullet & \bullet \\
	\bullet & \bullet & \bullet & \bullet
	\arrow[tail reversed, from=1-1, to=1-2]
	\arrow[tail reversed, from=1-2, to=1-3]
	\arrow[tail reversed, from=1-2, to=2-2]
	\arrow[tail reversed, from=1-4, to=1-3]
	\arrow[tail reversed, from=2-1, to=1-1]
	\arrow[tail reversed, from=2-2, to=2-1]
	\arrow[tail reversed, from=2-2, to=3-2]
	\arrow[tail reversed, from=2-3, to=1-3]
	\arrow[tail reversed, from=2-3, to=2-2]
	\arrow[tail reversed, from=2-3, to=2-4]
	\arrow[tail reversed, from=2-4, to=1-4]
	\arrow[tail reversed, from=3-1, to=2-1]
	\arrow[tail reversed, from=3-1, to=3-2]
	\arrow[tail reversed, from=3-2, to=3-3]
	\arrow[tail reversed, from=3-3, to=2-3]
	\arrow[tail reversed, from=3-3, to=3-4]
	\arrow[tail reversed, from=3-4, to=2-4]
\end{tikzcd}

    \caption{A $3\times 4$ grid-graph}
    \label{fig:placeholder}
\end{figure}
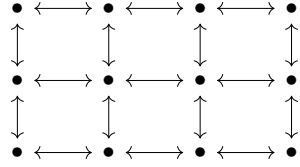
\end{definition}
\vspace{-2em}
The definition of a grid graph makes the underlying structure and adjacency relations entirely explicit. 
The formulation is intended to suggest a natural generalization to scene graphs~\cite{chang2021comprehensive}.
In scene graphs, vertices and edges capture objects in the environment and relationships which hold between them.

We support temporal reasoning via traces: sequences in which each element is associated with a point in time.
Each element of the sequence is a pair of functions assigning, at that time point, propositional variables and nominals to cells in the grid-graph.
\begin{definition}[States \& traces]\rm
	Given a grid-graph $G = (\Pos, E),$ a set $\AP$ of atomic propositions, a set $\NP$ of nominals, and a natural number $n \in \mathbb{N},$ a \emph{state} $t = (x,y)$ is a pair of valuation maps, where $x : \AP \to 2^{\Pos}$ maps each proposition to a subset of points in $G$ and $y : \NP \to \Pos$ maps each nominal to a unique point in $G$. 
    A \emph{trace} $\vec{t} = \big(t_0, \ldots, t_n\big)$ of length $n+1$ consists of a finite sequence of states.

\end{definition}
Let $\vec{t}^k \defeq \big((x_k,y_k), (x_{k+1},y_{k+1}),\dots, (x_n,y_n)\big )$ denote the suffix of $\vec{t}$ starting at the $k$-th coordinate;
and let $\subst{\vec{t}}{v}{p}$ be the element-wise substitution operation on the nominal component of the states, i.e.\ $$\subst{\vec{t}}{v}{p} \defeq \big((x_0,\subst{y_0}{v}{p}), \dots, (x_n,\subst{y_n}{v}{p})\big).$$ 
Putting together grid-graphs and traces, we obtain our semantics for HSTL.
\begin{definition}[Semantics]\rm
    Given a grid-graph $G=(\Pos,E)$, a trace $\vec{t} = \big((x_0,y_0),\dots, (x_n,y_n)\big)$ of length $n+1$ and a point $p_{i,j} \in \Pos$, the satisfaction relation $\models$ is defined as follows:
    \begin{itemize}
        \item $G,\vec{t}, p_{i,j} \models \top$
        \item $G,\vec{t}, p_{i,j} \models a$ iff $p_{i,j} \in x_0(a), \text{ for }a \in \AP$
        \item $G,\vec{t}, p_{i,j} \models v$ iff $p_{i,j} = y_0(v), \text{ for }v \in \NP$
        \item $G,\vec{t}, p_{i,j} \models \neg \varphi$ iff $G,\vec{t}, p_{i,j} \not \models \varphi$
        \item $G,\vec{t}, p_{i,j} \models \varphi_1 \land \varphi_2$ iff $G,\vec{t}, p_{i,j} \models \varphi_1$ and $G,\vec{t}, p_{i,j} \models \varphi_2$
        \end{itemize}
        Temporal modalities
        \begin{itemize}
          \item   $G,\vec{t}, p_{i,j} \models \opX \varphi$ iff  $\vec{t}^1$ is well-defined and $G,\vec{t}^1, p_{i,j} \models \varphi$ 
          \item   $G,\vec{t}, p_{i,j} \models \varphi_1 \opU \varphi_2$ iff $\exists k \in \{0,\dots,n\}$ s.t. \\  $\vec{t}^k$ is well-defined, $G,\vec{t}^k, p_{i,j} \models \varphi_2,$  and  $G,\vec{t}^l, p_{i,j} \models \varphi_1,  \text{ for all } \;  l <k.$
    \end{itemize}
    Spatial modalities
    \begin{itemize}
        \item  $G,\vec{t}, p_{i,j} \models \opF \varphi$ iff $p_{i+1,j}\in \Pos$ and $G,\vec{t}, p_{i+1,j} \models \varphi$
         \item  $G,\vec{t}, p_{i,j} \models \opB \varphi$ iff $p_{i-1,j}\in \Pos$ and $G,\vec{t}, p_{i-1,j} \models \varphi$
         \item  $G,\vec{t}, p_{i,j} \models \opL \varphi$ iff $p_{i,j-1}\in \Pos$ and $G,\vec{t}, p_{i,j-1} \models \varphi$
          \item  $G,\vec{t}, p_{i,j} \models \opR \varphi$ iff $p_{i,j+1}\in \Pos$ and $G,\vec{t}, p_{i,j+1} \models \varphi$
          \end{itemize}
          Hybrid modalities
          \begin{itemize}
        \item $G,\vec{t}, p_{i,j} \models \opAt{v} \varphi$ iff $G,\vec{t}, y_0(v) \models \varphi$
        \item $G,\vec{t}, p_{i,j} \models \opBind{v} \varphi$ iff $G,\subst{\vec{t}}{v}{p_{i,j}}, p_{i,j} \models \varphi$
    \end{itemize}
We write $\models \varphi$ to denote the validity of a formula $\phi$, i.e.\ $G,\vec{t},p_{i,j}\models \varphi$ for all grid-graphs $G$, traces $\vec{t}$ and points $p_{i,j}$.
We write $\not\models\varphi$ for $\varphi$ that is not valid.
\end{definition}

\paragraph*{Derived operators.} 
Let  $D \in \{\opL,\opR,\opF,\opB\}$ be any direction.
We define $\langle D\rangle^n \psi\defeq \bigvee_{i=1}^n D^i (\psi)$ and $[ D ]^n \psi\defeq\bigwedge_{i=1}^n (D^i(\top) \to D^i(\psi ))$, where $D^i$ denotes $i$ copies of nested $D$ modalities. 
These respectively express that $\psi$ is satisfied in some or every of the (extant) cells in direction $D$ within distance $n$.
When discussing a fixed grid-graph, we write $\langle D \rangle$ and $[ D ]$ to mean that the omitted superscript $n$ is the number of rows in the grid-graph when $D\in \{ \opF,\opB\}$ or the number of columns when $D\in \{\opL,\opR\}$.

Global and future modalities $\opBox{}\varphi$ and $\opDiamond{}\varphi$ respectively mean $\varphi$ is true at all or some future (or current) times. Definitions $\opDiamond{} \varphi \defeq \top \opU \varphi$ and $\opBox{} \varphi \defeq \neg \opDiamond{} \neg \varphi$ are standard.
We define a ``weak next'' modality $\opWX \varphi \defeq (\opX \top) \to (\opX \varphi)$, which is like $\opX$ but is true during the last timestep.
Note, $\opWX{}$ is the dual of $\opX$. 

\begin{example}[Temporal modalities]
	Let $\safe\in \AP$, $\POV\in \NP$ and consider 
	$ (\safe) \opU (\POV).$
	Evaluated at $\vec{t},p_{i,j}$, it expresses that the point $p_{i,j}$ is safe, until at some point in time $\POV$ occupies $p_{i,j}$.
\end{example}
\begin{example}[Hybrid modalities] Let $\SV,v \in \NP$ and consider 
$\opAt{\SV}  \opBind{v}  \bigcirc \opAt{\SV} v$. 
	Evaluated at $\vec{t},p_{i,j}$ it expresses that the next time step exists, and that $\SV$ remains in the same position at the next time step. 
\end{example}

\subsection{Nontrivial validities}
\label{subsec:valities}
We explore the following valid formulas of HSTL to validate that HSTL semantics align with intuition as desired.
Since our frames are rectangular $m\times n$ grids with temporal traces, the interaction between spatial modalities, temporal operators, and hybrid operators yields a number of nontrivial validities.

From a logical perspective, these validities serve as a sanity check that the operators behave according to the intended intuitions. 
Establishing validities ensures that the formal semantics aligns with the conceptual reading of the operators.
\begin{enumerate}
	\item Commutation of orthogonal moves: horizontal and vertical moves commute on every grid point.
Thus, for all formulas $\varphi$,
$$
\opF\opR\varphi\leftrightarrow\opR\opF\varphi,
\hspace{1em}
\opF\opL\varphi\leftrightarrow\opL\opF\varphi,
$$
$$
\opB\opR\varphi\leftrightarrow\opR\opB\varphi,
\hspace{1em}
\opB\opL\varphi\leftrightarrow\opL\opB\varphi.
$$
These express that moving ``up and then right'' reaches the same point as ``right and then up'', and similarly for all orthogonal pairs.
\item Loops on the grid: traversing the four sides of a unit square and returning to the start yields the original truth value, whenever the entire cycle exists:
$$
\opF\opR\opB\opL\varphi\rightarrow\varphi,
\hspace{1em}
\opR\opF\opL\opB\varphi\rightarrow\varphi,
$$
$$
\opF\opL\opB\opR\varphi\rightarrow\varphi,
\hspace{1em}
\opB\opR\opF\opL\varphi\rightarrow\varphi.
$$
\item Space and time interaction:
spatial modalities modify the grid coordinate while temporal modalities modify the trace index. 
Hence they commute.
\begin{enumerate}
\item Commutation of temporal with spatial moves: for all $\varphi$,
$$
\opX D\varphi\leftrightarrow D\opX\varphi, 
\hspace{1em}
\lozenge D\varphi\leftrightarrow D \lozenge \varphi, 
\hspace{1em} 
\square D\varphi\leftrightarrow D \square \varphi. 
$$
\item  Distribution of spatial moves along until: for all $\varphi$ and $\psi$,
$$
D(\varphi\opU\psi)\leftrightarrow(D\varphi)\opU(D\psi).
$$
\end{enumerate}

\item Hybrid validities: 
The semantics of the hybrid modalities interacts nontrivially with space and time.
The following standard~\cite{brauner2010hybrid} hybrid validities hold for $a,b,c \in \NP$:
$$ \opBind{a} a, 
\hspace{1em} \opAt{a} a, \hspace{1em} \opAt{a} b \rightarrow \opAt{b} a, \hspace{1em} \opAt{a} b \wedge \opAt{a} \varphi \rightarrow \opAt{b} \varphi $$

As for binders, we have the following list of  validities for all $v,\varphi,\psi$:
$$
\opBind{v}\varphi\leftrightarrow\opBind{v}\opAt{v}\varphi, 
\hspace{1em} 
\opBind{v}\opX\varphi\leftrightarrow\opX\opBind{v}\varphi, 
\hspace{1em} 
\opBind{v}\lozenge\varphi\leftrightarrow\lozenge\opBind{v}\varphi, 
$$
$$
\opBind{v}\square\varphi\leftrightarrow\square\opBind{v}\varphi, \hspace{1em}  \opBind{v}(\varphi\opU\psi)\leftrightarrow(\opBind{v}\varphi)\opU(\opBind{v}\psi).
$$
\end{enumerate}
\subsection{Nontrivial non-validities}
\label{subsec:non-valities}
In contrast to the validities above, the following formulas
fail on some model in the class of HSTL models. 
Each failure witnesses a different characteristic of the interaction between spatial, temporal and hybrid operations. 
Exploring these non-validities provides an awareness of where formal semantics may differ from intuition and helps establish the limits of HSTL expressiveness.
\begin{enumerate}
	\item The hybrid operator $\opAt{}$ is time sensitive, i.e.\ 
$
\not\models\opAt{v}\varphi\rightarrow\opX\opAt{v}\varphi.
$
	\item The hybrid operator $\opAt{}$ and $\lozenge$ do not commute, i.e.\ 
	$
\not\models\opAt{v}\lozenge\varphi\leftrightarrow\lozenge\opAt{v}\varphi.$
	\item Hybrid operators and $D$ modalities do not commute, i.e.\
	$\not\models\opAt{v}D\varphi\leftrightarrow D\opAt{v}\varphi$ and $\not\models\opBind{v}D\varphi\leftrightarrow D\opBind{v}\varphi.$
\end{enumerate}

\section{Modeling}
The simple operators of HSTL combine to express rich specifications of how vehicle motion evolves over time in a dynamic road environment.
We demonstrate how the HSTL operators work together to model an automotive scenario. 
Interactions are complex, but most practical models build on a few fixed idioms:
\begin{definition}[Modeling idioms]\rm
\label{def:model-idioms}
\begin{enumerate} 
    \item A \emph{global state assumption} assumes a formula at all times for all initial viewpoints $p_{i,j}$.
    Specifically we support  $\square \opAt{v} \varphi$ where $\varphi$ contains no temporal modalities other than $\square$.
    Clearly $\square$ captures all times, but $\opAt{v}$ capturing all initial viewpoints is more subtle: 
    $\opAt{v}$ overwrites the viewpoint with  $v's$ position, so its truth value is viewpoint-independent.
    This detail simplifies model-checker soundness in certain edge cases.
    \item A \emph{static car assumption} consists of a single nominal. 
    The assumption $v \in \NP$ specifies that $v$ is stationary for the whole trace, i.e., $y_k(v) = p_{i,j}$ for fixed values of $i,j$ and any time $k$ in the trace.
    It is expressed by the formula  $\opAt{v} \opBind{v'} \square @_{v} v'.$
    We say that $v$ is \emph{static}.
    \item A \emph{relative motion assumption} consists of a pair of nominals $v_1, v_2$, and a sequence of spatial modalities $\vec{D}$. 
    The assumption $(v_1, v_2, \vec{D})$ specifies that relative to $v_1$, $v_2$ is always in the position obtained by moving according to $\vec{D}$.
    It is expressed by the formula  $\square @_{v_1} \mathop{\vec{D}} v_2.$
    We say that $v_1$ is a \emph{dependee} and that $v_2$ is \emph{dependent} on $v_1$.
    \item A \emph{fixed motion assumption} consists of a nominal $v$ and a set of sequences of spatial modalities $M$.
    The assumption $(v, M)$ specifies that after one time step, $v$'s original position relates to its new position by one of the directions in $M$.
    It is expressed by  $\square \opAt{v} \opBind{v'} \opWX \opAt{v} \bigvee_{\vec{D} \in M} \mathop{\vec{D}} v'.$
    We say that $v$ has \emph{fixed motion} from $M$. 
\end{enumerate} 
\end{definition}

The following examples, which use a grid-graph of size $5 \times 3,$ build upon the idioms.
Not all idioms appear directly in these examples, but all are supported in our model-checker (Section \ref{sec:model-checking}).
In Figures~\ref{fig:safefollow} and \ref{fig:evading}, 
the red car denotes $\SV$, the blue car denotes $\POV$, and the proposition
$h$ (road hazard) is illustrated by barricades. 
The yellow arrows are not part of the model; they merely
indicate, for convenience, how the location of a nominal changes at the
next time step.
\begin{example}[Safe at all times] The \emph{global state} formula
$   \varphi_1=\square \opAt{\SV} \neg\POV$
expresses that, at all points in time (i.e.\ at every state in $\vec{t}$), $\SV$ always occupies a different cell from $\POV$, and is therefore safe.
\end{example}
\begin{example}[Safe follow]
\label{ex:safe-follow}
The formula 
    	\begin{align*}
		\varphi_2 := \opAt{\SV}&\neg \opB\top\wedge  \square \big( \opAt{\POV}\opBind{v'}\opWX \opAt{\POV}(v' \vee \opB(v'))\big) \\&\wedge \square  \Big(\opAt{\SV}\opBind{v}\opWX \opAt{\SV}\big( (\neg \POV \wedge \opB (v)) \vee  (v \wedge \opF (\POV))\big) \Big)
	\end{align*} 
expresses that $\SV$ follows $\POV$ safely. 
The first conjunct specifies that $\SV$ starts in the bottom row of the grid-graph, while the second conjunct, a \emph{fixed motion assumption}, requires that at each time step $\POV$ either moves one cell forward or remains in place. 
To match this behavior, the third conjunct enforces that $\SV$ advances whenever the cell directly ahead is not occupied by $\POV$, and otherwise stays in place.
A model satisfying $\varphi_2$ is illustrated in Figure~\ref{fig:safefollow}. 
\begin{figure}[H]
  \centering
  \includegraphics[width=0.18\textwidth]{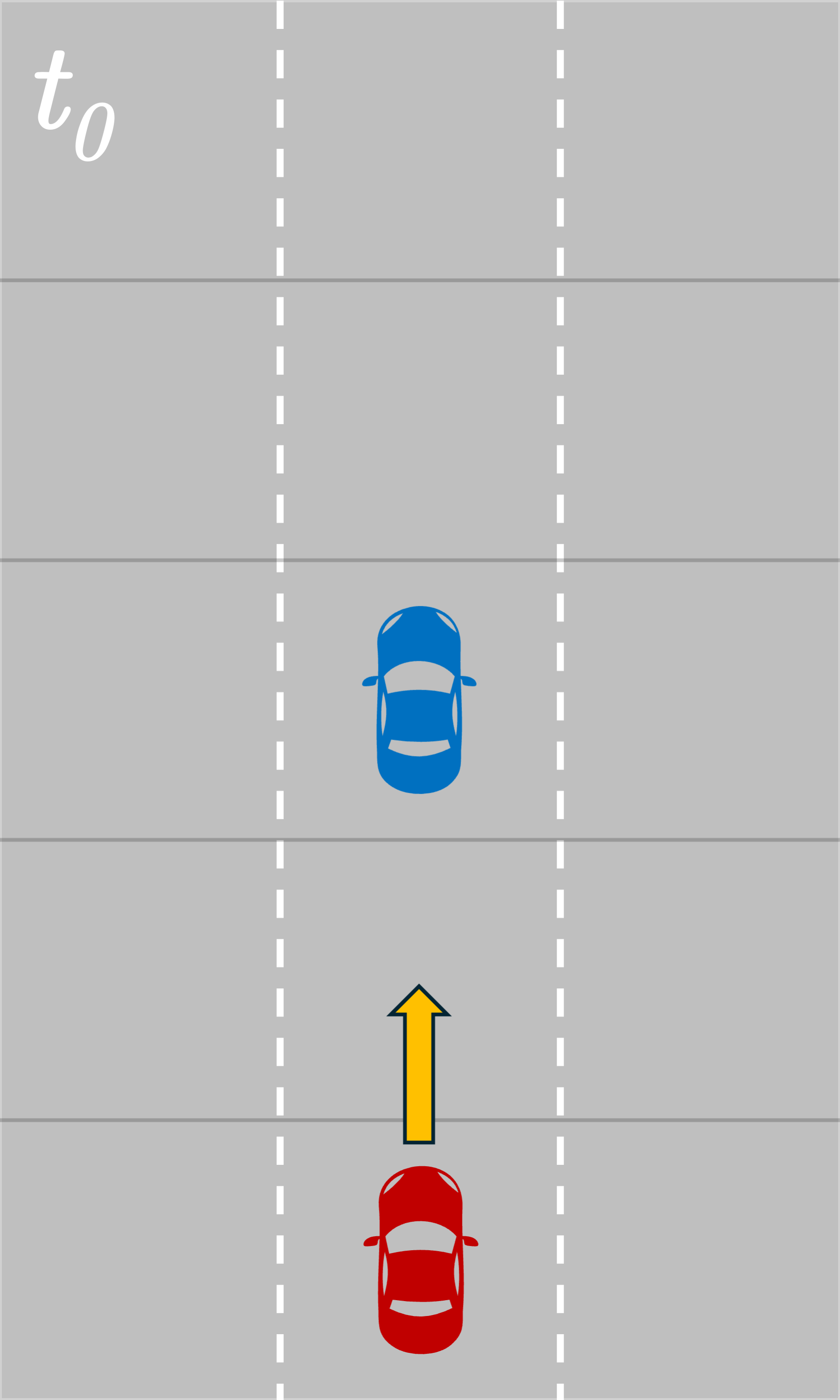}
  \hspace{1.5em}
  \includegraphics[width=0.18\textwidth]{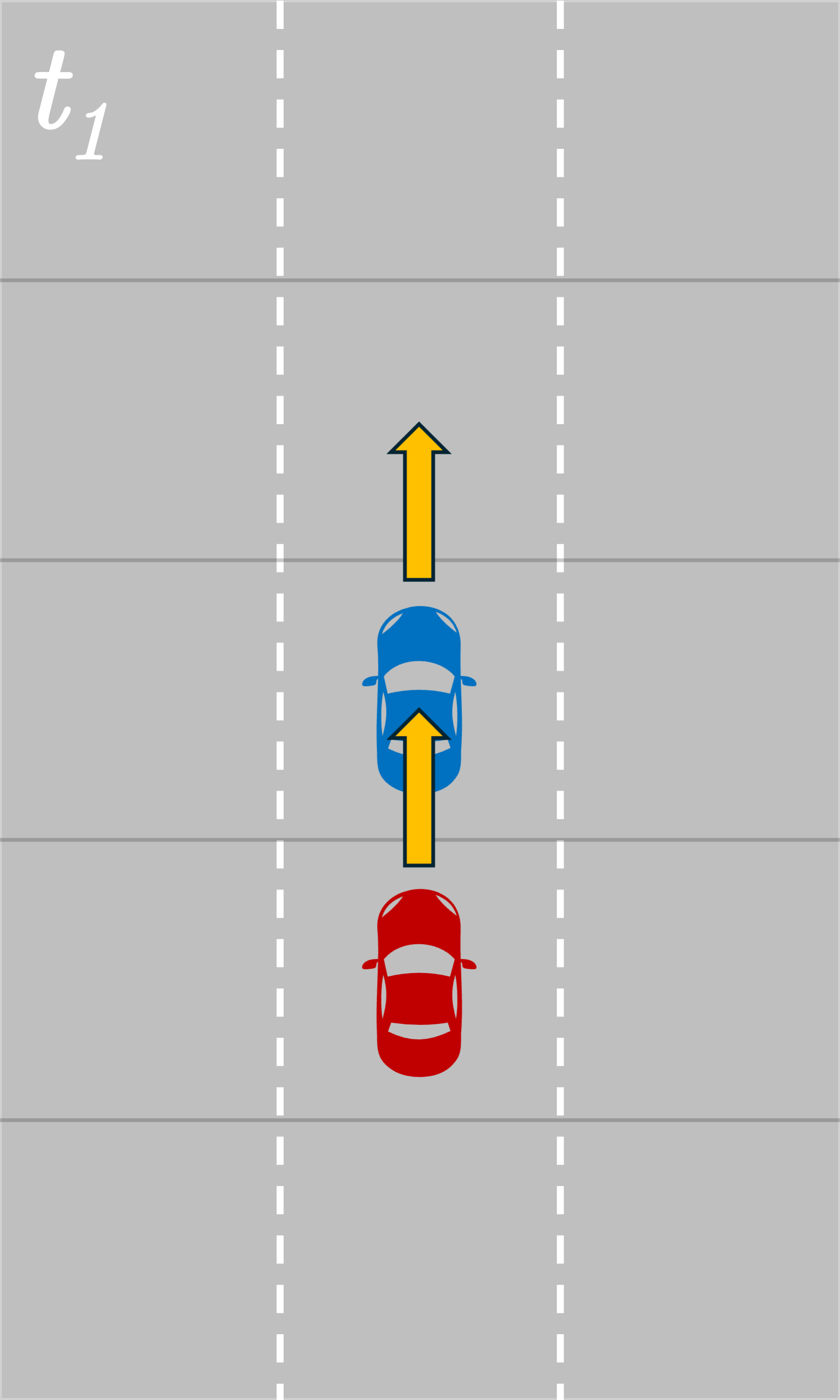}
  \hspace{1.5em}
  \includegraphics[width=0.18\textwidth]{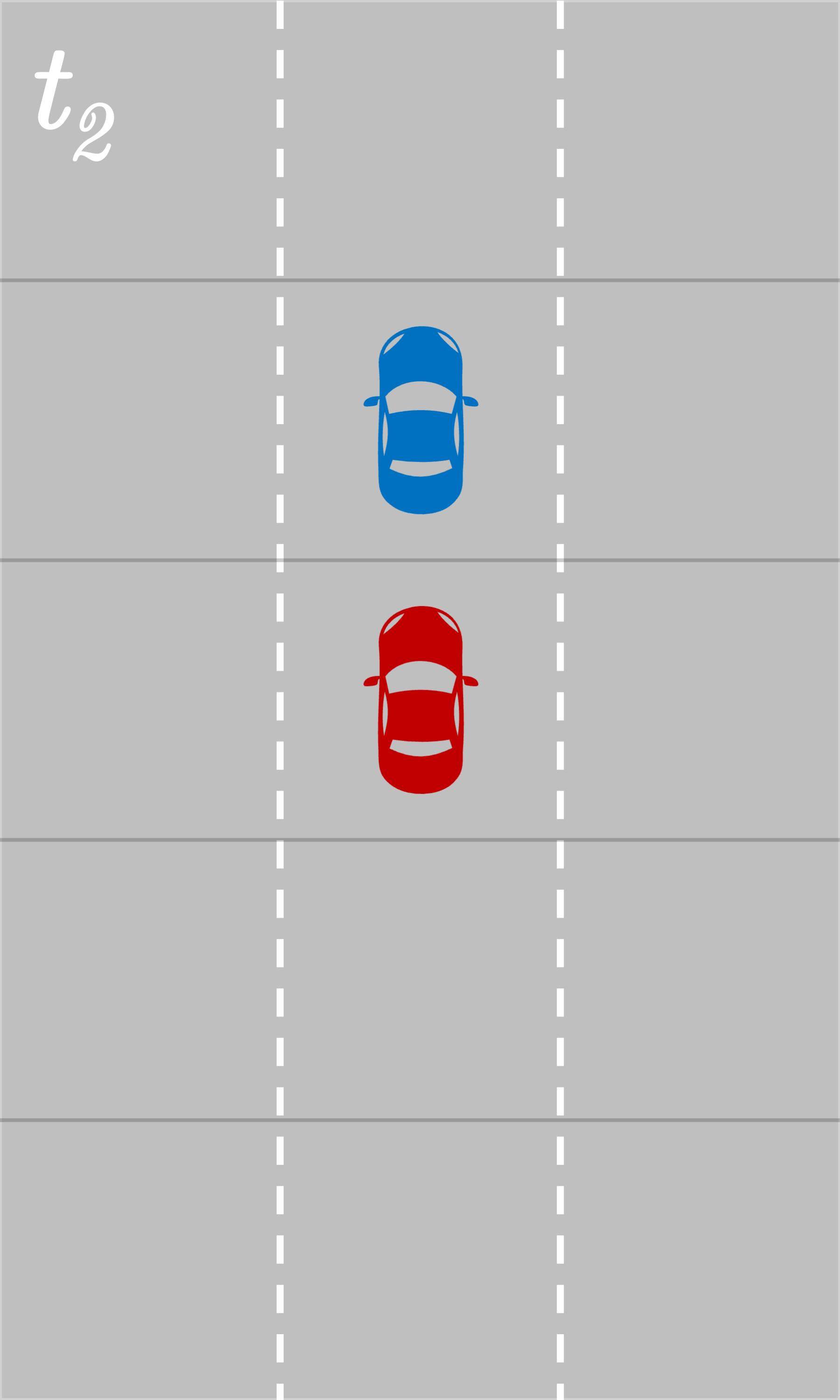}
  \caption{A model satisfying the formula $\varphi_2$}
  \label{fig:safefollow}
\end{figure}
\end{example}
\begin{example}[Evading static hazard on an active road]
\label{ex:static-hazard}
The formula 
 \begin{align*}
      \varphi_3:=   \opAt{\SV} \Big( &\big(\opR(\POV)\wedge \langle\opF\rangle\square h \big) \wedge \\ &\big(\opAt{SV}\opBind{v}\opX \opAt{\SV}(\opB(v)\wedge\square\neg h) \big)\opU \phantom{a} \\ & \big(\opAt{SV} \opBind{v} \opX \opAt{\SV} (\opL(v) \wedge \langle\opF\rangle(\POV)\wedge [\opF]\square\neg h)\big)\Big)
     \end{align*}
expresses the following behavior: 
$\POV$ is initially to the right of $\SV$ and there is a road hazard ahead whose position does not change over time.
In addition, $\SV$ must eventually move to a position in the neighboring right-hand lane such that $\POV$ is in front of it and no road hazard lies ahead. 
Until that moment, $\SV$ remains in the lane of its initial position while avoiding collisions with road hazards. 
A model satisfying $\varphi_3$ is illustrated in Figure~\ref{fig:evading}. 
   \begin{figure}[H]
  \centering
  \includegraphics[width=0.18\textwidth]{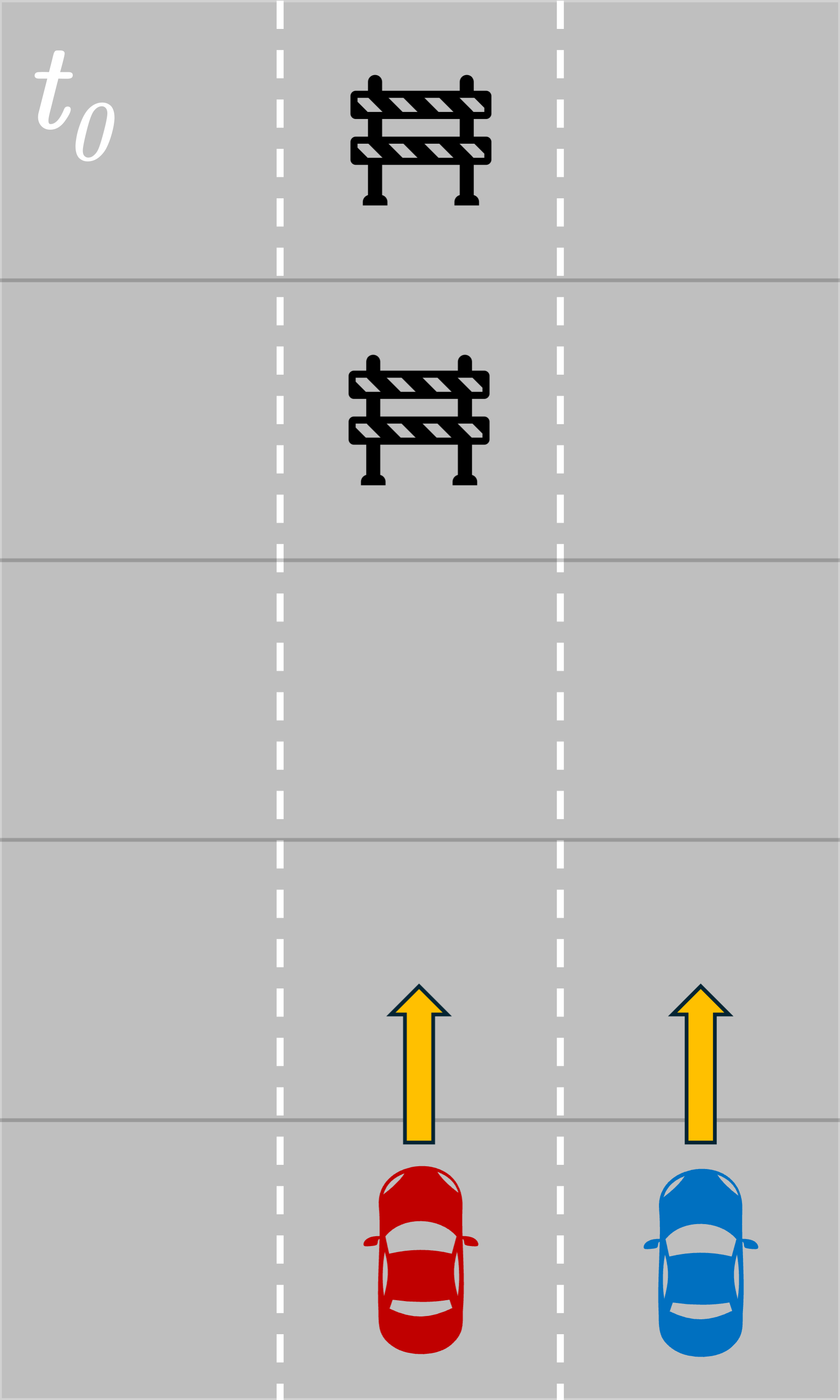}
  \hspace{1.5em}
  \includegraphics[width=0.18\textwidth]{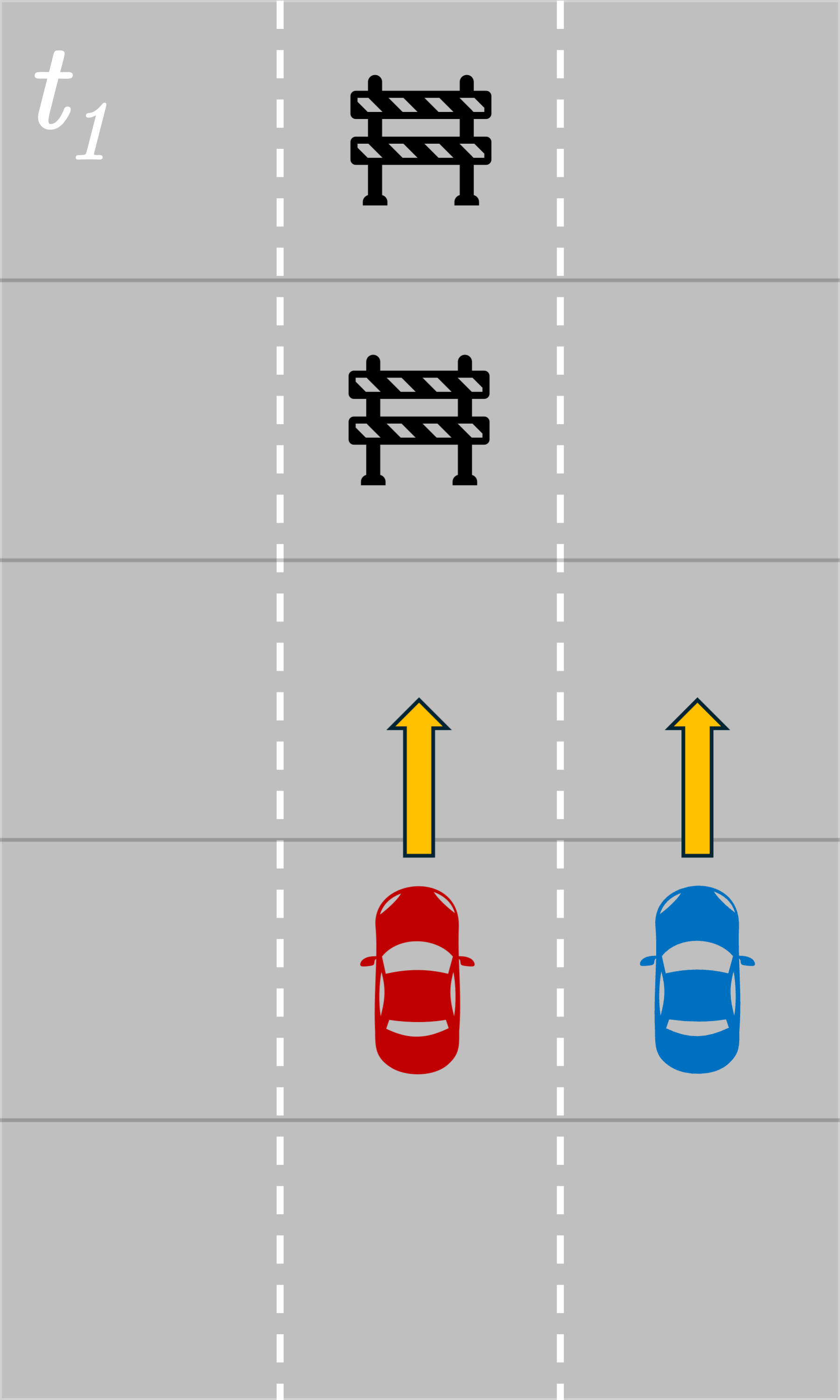}
  \hspace{1.5em}
  \includegraphics[width=0.18\textwidth]{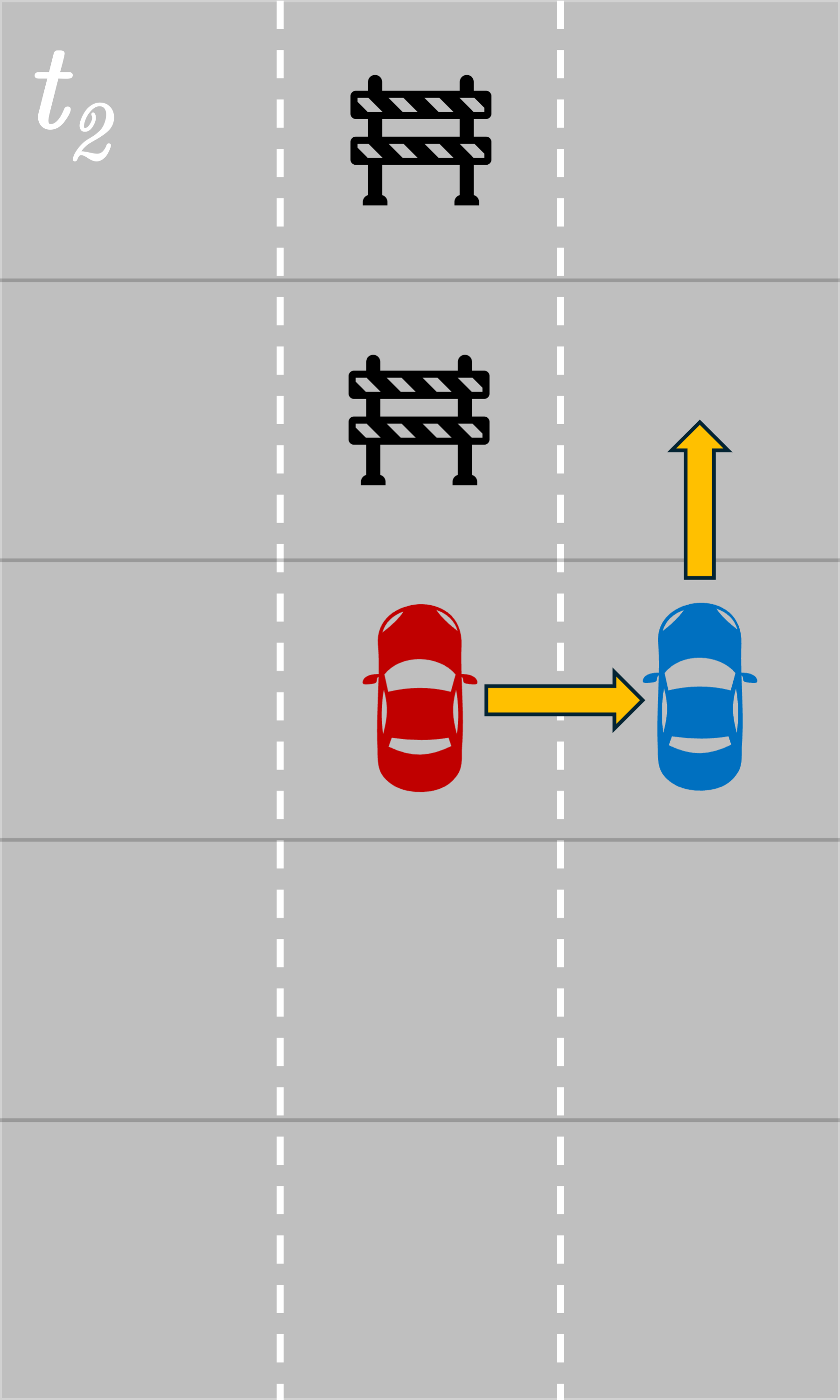}
   \hspace{1.5em}
  \includegraphics[width=0.18\textwidth]{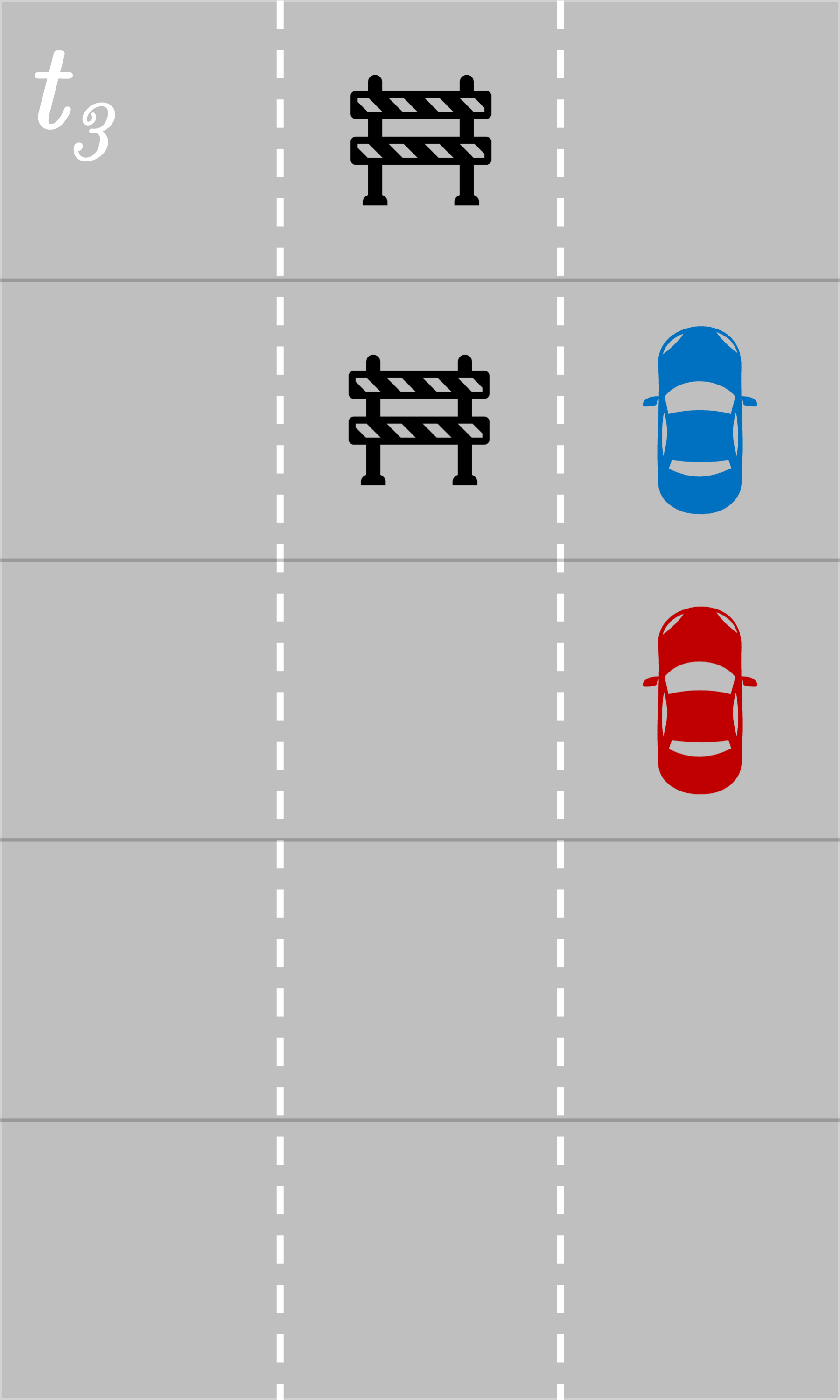}
  \caption{A model satisfying the formula $\varphi_3$}
  \label{fig:evading}
\end{figure}
Note that Figure \ref{fig:evading} is not the only model of $\varphi_3$.
Line 2 allows $\POV$ to be anywhere, and the value of $h$ is arbitrary everywhere except the current position and its front region.
This under-specification reduces realism \emph{and} makes the search space intractably large.
Thus $\varphi_3$ showcases our motivations for nominals and motion assumptions: increased realism and reduced search space.
\end{example}

\section{Model-Checking}
\label{sec:model-checking}
Model-checking is a broad set of verification techniques that systematically explore system states to verify correctness in each state.
We develop model-checking algorithms for HSTL, which systematically explore possible traces and identify those which satisfy a specification.
This algorithm is appropriate for offline usage, e.g., for test case generation.
Additionally, we give correctness and complexity guarantees (Theorems~\ref{thm:correctness and runtime of evaluator} and \ref{thm:motion algorithm number of traces generated}).

\subsection{Formula evaluation}
The formula evaluation algorithm takes as input a grid-graph $G$, a trace $\vec{t}$, a point $p_{i,j}$ in the grid-graph, and an HSTL formula $\varphi$, and outputs whether or not $G, \vec{t}, p_{i,j} \models \varphi$. 
Its pseudocode is given in Algorithm~\ref{alg:evaluator}. 
It adapts a modern dynamic programming algorithm from LTL~\cite{DBLP:journals/jair/FiondaG18} to avoid an exponential dependence on the length of the formula that can arise from a na\"ive recursive evaluation of the Until operator. 
\paragraph{Remark.}
In the implementation of the lookup table \texttt{memo}, distinct instances of a subformula $\varphi'$ of $\varphi$ result in distinct keys for the lookup table. 
That is, in the execution of \texttt{\textsc{Eval\_memo}}, if $\varphi'$ and $\varphi''$ are syntactically identical formulas that arise separately in the parse tree of $\varphi$, the keys $(t, p_{i,j}, \varphi')$ and $(t, p_{i,j}, \varphi'')$ are considered distinct.
This is necessary to ensure correctness of the spatial operators and the bind operator, since they alter the current position and the trace, respectively, on which the formula is evaluated. 

\begin{algorithm}
    \caption{Formula evaluator}\label{alg:evaluator}
    \begin{algorithmic}[1]
        \Function{\texttt{Eval}}{$G, \vec{t}, p_{i,j}, \varphi$}
            \State \Input Grid-graph $G = (\Pos, E)$, trace $\vec{t} = ((x_0, y_0), \dots, (x_n,y_n))$, 
            \State \phantom{\Input}point $p_{i,j}$, HSTL formula $\varphi$
            \State \Output Bool
            \State \textbf{initialize} lookup table \texttt{memo}
            \State \Return \texttt{\textsc{Eval\_memo}}($G, \vec{t}, 0, p_{i,j}, \varphi$)
        \EndFunction

        \State
        
        \Function{\texttt{\textsc{Eval\_memo}}}{$G, \vec{t}, k, p_{i,j}, \varphi$}
            \State \Input Grid-graph $G = (\Pos, E)$, trace $\vec{t} = ((x_0, y_0), \dots, (x_n,y_n))$, 
            \State \phantom{\Input}time step $k \leq n$, point $p_{i,j}$, HSTL formula $\varphi$
            \State \Output Bool
            \If{\texttt{memo}$[(k, \varphi)]$ exists} \Return{\texttt{memo}$[(k,  \varphi)]$}
            \Else
                \Match{$\varphi$}
                    \Case{$\top$} \texttt{result} $\gets$  True \EndCase
                    \Case{$a \in \AP$} \texttt{result} $\gets$  $p_{i,j} \in x_0(a)$ \EndCase
                    \Case{$v \in \NP$} \texttt{result} $\gets$  $p_{i,j} = y_0(v)$ \EndCase
                    \Case{$\neg \varphi_1$} \texttt{result} $\gets$ not \texttt{\textsc{Eval\_memo}}$(G, \vec{t}, k, p_{i,j}, \varphi_1)$ \EndCase
                    \Case{$\varphi_1 \land \varphi_2$} 
                        \State \texttt{result} $\gets$ \texttt{\textsc{Eval\_memo}}$(G, \vec{t}, k, p_{i,j}, \varphi_1)$ and \texttt{\textsc{Eval\_memo}}$(G, \vec{t}, k, p_{i,j}, \varphi_2)$
                    \EndCase
                    \Case{$\opX \varphi_1$} 
                        \State \texttt{result} $\gets$ $k < n$ and \texttt{\textsc{Eval\_memo}$(G, \vec{t}, k+1, p_{i,j}, \varphi_1)$}
                    \EndCase
                    \Case{$\varphi_1 \opU \varphi_2$} 
                        \If{\texttt{\textsc{Eval\_memo}$(G, \vec{t}, k, p_{i,j}, \varphi_2)$}} \texttt{result} $\gets$ True \label{line:evaluator:until, varphi2 true}
                        \ElsIf{$k < n$}
                            \State \texttt{result} $\gets$ \texttt{\textsc{Eval\_memo}}$(G, \vec{t}, k, p_{i,j}, \varphi_1)$ and 
                            \State \phantom{\texttt{result} $\gets$} \texttt{\textsc{Eval\_memo}}$(G, \vec{t}, k+1, p_{i,j}, \varphi)$
                        \Else{} \texttt{result} $\gets$ False
                        \EndIf
                    \EndCase
                    \Case{$\opF \varphi_1/\opB \varphi_1/\opR \varphi_1/\opL \varphi_1$}
                        \If{$p_{i+1,j}/p_{i-1,j}/p_{i,j+1}/p_{i,j-1} \in \Pos$}
                            \State $p' \gets p_{i+1,j}/p_{i-1,j}/p_{i,j+1}/p_{i,j-1}$ \label{line:evaluator:spatial update}
                            \State \texttt{result} $\gets$  \texttt{\textsc{Eval\_memo}$(G, \vec{t}, k, p', \varphi_1)$} 
                        \Else{} \texttt{result} $\gets$ False
                        \EndIf
                    \EndCase 
                    \Case{$\opAt{v} \varphi_1$} \texttt{result} $\gets$  $\texttt{\textsc{Eval\_memo}}(G, \vec{t}, k, x_0(v), \varphi_1)$ \EndCase
                    \Case{$\opBind{v} \varphi_1$} 
                        \State $\vec{t}' = ((x'_0, y'_0), \dots, (x'_n,y'_n)) \gets$ \texttt{copy}($\vec{t}$)
                        \For{$l = k, \ldots, n$}
                            \State $y_l'(v) \gets p_{i,j}$ 
                        \EndFor \label{line:evaluator:bind reassignment complete}
                        \State \texttt{result} $\gets$  $\texttt{\textsc{Eval\_memo}}(G, \vec{t}', k, p_{i,j}, \varphi_1)$
                    \EndCase
                \EndMatch
                \State \texttt{memo}$[(k, \varphi)] \gets$ \texttt{result} \label{line:evaluator:save result to memo}
                \State \Return{\texttt{result}}
            \EndIf
        \EndFunction
    \end{algorithmic}
\end{algorithm}

\begin{theorem}[Correctness and runtime of formula evaluator] \label{thm:correctness and runtime of evaluator}
    For every grid-graph $G$, trace $\vec{t}$, point $p_{i,j}$ in the grid-graph, and HSTL formula $\varphi$, \linebreak
    $\texttt{Eval}(G, \vec{t}, p_{i,j}, \varphi)$ (as defined in Algorithm~\ref{alg:evaluator}) returns True if and only if \linebreak
    $G, \vec{t}, p_{i,j} \models \varphi$.
    Furthermore, \texttt{Eval} terminates within $O(|\vec{t}|^2 \cdot |\varphi|)$ steps.
\end{theorem}

\subsection{Model-checkers}
The model-checkers take as input a grid-graph $G$, a formula $\varphi$, a max trace length $n$, and optionally a list of assumptions $A$, in the form of the modeling idioms (Definition~\ref{def:model-idioms}).
They output traces $\vec{t}$ of length up to $n$ which satisfy the assumptions $A$, for which there are points $p_{i,j}$ such that $G, \vec{t}, p_{i,j} \models \varphi$. 
Three different model-checking algorithms are implemented:
\begin{enumerate}
    \item The \emph{baseline} algorithm generates all possible states on $G$ (i.e., arbitrary assignments of atomic propositions and nominals), then generates all possible traces by arbitrarily assigning a state to each time step.
    Once generated, the specification is evaluated on each trace. 
    \item The \emph{optimized} algorithm allows for global state assumptions, as formalized in Definition \ref{def:model-idioms}. The algorithm begins by identifying the set of states $S$ that satisfy the given set of assumptions $A$. More precisely, it selects all states $t = (x,y)$ such that $G, t, p \models \alpha$ for every $p \in Pos$ and $\alpha \in A$. It then generates, similarly to the baseline algorithm, all possible traces composed exclusively of states in $S$. Once generated, the specification is evaluated on each trace.
    \item The \emph{motion} algorithm additionally allows motion assumptions.
    Traces are generated one step at a time, utilizing the motion assumptions to limit the number of assignments of nominals that need to be considered.
    For example, for a fixed motion assumption of the form $(v, M)$, only states in which $v$ has moved from its previous position according to $M$ need to be considered.
    Once a candidate state for the next timestep is generated, it is evaluated against the global state assumptions, and kept iff it satisfies them.
    Traces generated this way are finally checked against the entire specification. 
\end{enumerate}

The pseudocode for the motion algorithm is given in Algorithm~\ref{alg:motion algorithm}.
See Definition~\ref{def:model-idioms} for the motion assumption templates, which are central to its design.
Furthermore, we make these assumptions about the inputs:
\begin{preconditions}[motion algorithm]
    \label{pre:preconditions for motion algorithm}
    We assume that the assumptions $A$ given as input to the motion algorithm satisfy the following:
    \begin{enumerate}
        \item $A$ is partitioned as $A = \A{static} \cup \A{rel} \cup \A{fixed} \cup \A{global}$ into static car, relative motion, fixed motion, and global state assumptions, respectively.
        \item $\A{static}, \A{rel}$, and $\A{fixed}$ are \emph{consistent}, that is, contain no contradictions.
        \item No nominal is both a dependee and dependent.
        \item Any dependent nominal appears in exactly one relative motion assumption.
    \end{enumerate}
\end{preconditions}
The last two conditions ensure that the positions of all dependent nominals can be determined precisely by (independently) choosing the positions of all dependee nominals.
We also describe several helper functions at high level, omitting implementation details:
\begin{itemize}
    \item \texttt{generate\_initial\_states}$(G, \A{rel}, \A{global})$ outputs the states which satisfy $\A{rel}$ and $\A{global}$. 
    That is, all dependent cars must be in the correct locations relative to their dependees, and every global assumption is satisfied.
    \item \texttt{valid\_dependee\_positions}$(G, \A{rel}, v)$ outputs positions where dependee nominal $v$ can be placed so all nominals dependent on $v$ stay in $G$.
    \item \texttt{valid\_prop\_assignments}$(G, \A{global})$ outputs the collection of assignments of atomic propositions in a single state which satisfy $\A{global}$.
    \item \texttt{complete\_state}$(G, \A{rel}, x, \tilde{y})$ takes as input an assignment of atomic propositions $x : \AP \to 2^{\Pos}$ and an assignment of all non-dependent nominals $\tilde{y} : \NP \rightharpoonup \Pos$, and outputs the state extending $(x, \tilde{y})$ in which all dependent nominals are assigned positions according to $\A{rel}$. 
    \item \texttt{check\_global\_assumptions}$(G, \vec{t}, \A{global})$ checks that all global state assumptions are satisfied by $\vec{t}$.
\end{itemize}

\begin{algorithm}
    \caption{Model-checker: motion algorithm} \label{alg:motion algorithm}
    \begin{algorithmic}[1]
        \Function{\texttt{Sat\_traces}}{$G, A, \varphi, n$}
            \State \Input Grid-graph $G = (Pos, E)$, assumptions $A$ satisfying preconditions~\ref{pre:preconditions for motion algorithm}, 
            \State \phantom{\Input} specification $\varphi$, max trace length $n$
            \State \Output Stream of pairs $(\vec{t}, \texttt{sat\_points}_{\vec{t}})$, where $\vec{t}$ is a trace and 
            \State \phantom{\Output} $\texttt{sat\_points}_{\vec{t}}$ is a subset of $Pos$
            \For{$\vec{t} \in \texttt{\textsc{Generate\_traces}}(G, A, n)$}
                \State \texttt{sat\_points} $\gets \{p \in \Pos \mid \texttt{\textsc{Eval}}(G, \vec{t}, p, \varphi) = \mathrm{True}\}$ 
                \If{\texttt{sat\_points} is not empty} \textbf{yield} $(\vec{t}, \texttt{sat\_points})$
                \EndIf
            \EndFor
        \EndFunction
        \State
        \Function{\texttt{Generate\_traces}}{$G, A, n$}
            \State \Input Grid-graph $G = (Pos, E)$, assumptions $A$, max trace length $n$
            \State \Output Stream of traces satisfying $A$
            \For{$\texttt{initial\_state} \in \texttt{generate\_initial\_states}(G, \A{rel}, \A{global})$}
                \State \textbf{yield} from \texttt{\textsc{Extend\_trace}}$(G, A, 1, n, \texttt{initial\_state}, [\texttt{initial\_state}])$ \label{line:motion alg:yield length 1 trace}
            \EndFor
        \EndFunction
        \State
        \Function{\texttt{Extend\_trace}}{$G, A, k, n, t, \vec{t}$}
            \State \Input Grid-graph $G = (Pos, E)$, assumptions $A$, current length $k$, 
            \State \phantom{\Input} max trace length $n$, previous state $t = (x,y)$, current trace $\vec{t}$ \qquad\quad\qquad(Note the font difference between $t$ and $\vec{t}$) 
            \State \Output Stream of traces extending $\vec{t}$ and satisfying $A$
            \State \textbf{yield} $\vec{t}$ \label{line:motion alg:trace generated}
            \If{$k = n$} \Return{}
            \Else{}
                \For{$v \in \NP$} \label{line:motion alg:nominal choice computation loop}
                    \If{$v$ is static} 
                        \State $C(v) \gets \{y(v)\}$ \label{line:motion alg:choices for static cars}
                    \ElsIf{$v$ has fixed motion from $M$} 
                        \State $C(v) \gets \{\mathop{\vec{D}} y(v) \in \Pos \mid \vec{D} \in M\}$ \label{line:motion alg:choices for fixed motion cars}
                    \ElsIf{$v$ is a dependee} 
                        \State $C(v) \gets$ \texttt{valid\_dependee\_positions}$(G, \A{rel}, v)$ \label{line:motion alg:choices for dependees}
                    \ElsIf{$v$ is not dependent} 
                        \State $C(v) \gets \Pos$ \label{line:motion alg:choices for free nominals}
                    \EndIf
                \EndFor
                \For{$\tilde{y} \in \prod_{v} C(v)$} \label{line:motion alg:nominal assignments loop}
                    \For{$x' \in \texttt{valid\_prop\_assignments}(G, \A{global})$}
                        \State $y' \gets \texttt{complete\_trace}(G, \A{rel}, x', \tilde{y})$ \label{line:motion alg:assign dependent nominals}
                        \State $t' \gets (x', y')$
                        \If{\texttt{check\_global\_assumptions}$(G, [t'], \A{global})$} \label{line:motion alg:check global assumptions}
                            \State \textbf{yield} from $\texttt{\textsc{Extend\_trace}}(G, A, k+1, n, t', \vec{t}\texttt{.append}(t'))$ \label{line:motion alg:generate next trace}
                        \EndIf
                    \EndFor
                \EndFor
            \EndIf
        \EndFunction
    \end{algorithmic}
\end{algorithm}

The main bottleneck in model-checking lies in generating the large number of traces required to test against the specification.
For simplicity of the analysis, suppose that the assignments of atomic propositions are empty at all times. 
If the grid has size $w \times h$, where $n$ is the maximum trace length and $m = |\NP|$, then naively generating all traces (as the baseline model-checker does) yields $\sum_{k=1}^n (w \times h)^{mk}$ traces: for each length $k$ up to $n$, all $m$ nominals must be assigned $k$ times, from $w \times h$ possibilities.
The optimized and motion algorithms improve on the baseline by limiting the number of traces that are generated based on the assumptions; therefore, it is necessary to show that these algorithms yield exactly those traces satisfying the assumptions and also give an upper bound on the number of traces generated. 

\begin{theorem}[Correctness and complexity of the motion algorithm] \label{thm:motion algorithm number of traces generated}
    Let $G$ be a grid graph of size $w \times h$, and let $A = \A{static} \cup \A{rel} \cup \A{fixed} \cup \A{global}$ be a set of assumptions satisfying preconditions~\ref{pre:preconditions for motion algorithm}.
    Additionally suppose that the assignments of atomic propositions are empty at all times.
    Let $s$ be the number of initial states which satisfy $\A{global}$.
    The branching factor $b_v$ of each $v$ is:
    \begin{itemize}
        \item $b_v = 1$ if $v$ is static or dependent,
        \item $b_v = |M|$ if $v$ has fixed motion from $M$,
        \item $b_v = |Pos| = w \times h$ otherwise.
    \end{itemize}
    Then \texttt{\textsc{Generate\_traces}}$(G, A, n)$ (defined in Algorithm~\ref{alg:motion algorithm}) yields exactly the \linebreak 
    traces which satisfy the assumptions $A$, and the number of traces checked is at most
    \[
        O \left( s \cdot \sum_{k=0}^{n-1} \left( \prod_{v \in \NP} b_v \right)^k \right).
    \]
\end{theorem}

The worst case for Algorithm~\ref{thm:motion algorithm number of traces generated} occurs when $s = (w \times h)^m$, i.e., there are no global state assumptions, and $b_v = w \times h$ for every $v \in \NP$, i.e., there are no motion assumptions. 
This aligns with the number of traces generated by the baseline checker. 
The addition of static and relative position assumptions improves performance as much as elimination of a nominal. 
For $k$-direction fixed motion, the speedup factor vs.\ baseline is $(w \times h)/k$. 
Such motion is fundamental to vehicle dynamics and appears in every test scenario, and therefore we expect the motion algorithm to provide significant improvements over the baseline.

\section{Evaluation}
\label{sec:evaluation}
To demonstrate the effectiveness of our model-checking algorithms, we evaluate the implementations on a variety of test cases.

\paragraph{Test Scenarios}
We evaluate the model-checking algorithms on seven test cases: two basic formulas and five test scenarios.
The first two are basic formulas of spatiotemporal and hybrid logics.
The next five, in order, are: safe following (Figure \ref{fig:safefollow}); hazard avoidance (Figure \ref{fig:evading}); safe intersection crossing where one road has priority; safe overtaking, in which a vehicle accelerates in a separate lane before merging back; and safe joining a platoon of vehicles in a neighboring lane.
These five scenarios are chosen because they are all fundamental and well-studied driving tasks; all except intersection crossing are focused on highway driving.
In modeling these scenarios, we support the claim that HSTL is flexible enough to model a wide range of concrete driving scenarios.
The test scenarios are designed for scalability: 
the grid size and trace duration can be increased to assess how well each algorithm handles increasing complexity. 
Platooning supports an arbitrary number of vehicles, allowing us to test scalability with respect to the number of nominals in a formula.

\paragraph{Test Results}
Experimental results are given in Table~\ref{tab:experimental-results}.
Tests were run on a workstation with 32 GB of RAM and a i7-14700 processor at 2.1 GHz.
A 10-minute timeout is used; timeouts are indicated with a horizontal line $-$. 
This timeout is appropriate because our algorithms are intended for offline use and a 10-minute timeout allows the test suite to complete within a day.
For online use, in contrast, 10 Hertz control loops are common, and planning timeouts would need to be consistent with such rates.
All reported times are wall-clock times measured in a multitasking environment, so variance is expected.

The first two tests are simple and thus run only once.
The test scenario parameters vary to explore scalability.
Safe following is single-lane, enabling longer lanes.
Safe crossing's grid size and crossing time are related, so we scale them in unison.
Safe overtaking increases duration while keeping grid size fixed.
Platooning's nominal count increases for fixed size and duration.

\begin{table}[!hbt]
\centering
\caption{Experimental results}
\begin{tabular}{ccccccccccc}
Test & Noms.\ & Grid & Len &  \#Sat & \#Trace1 & \#Trace2 & \#Trace3 & Time1 & Time2 & Time3 \\\hline
1     & 1            & (3,3)     & 3 & 819    &  819      & 819 & 819          & 0.0283   &  0.00417 & 0.0292\\\hline
2     & 2            & (3,3)     & 3 & 819    & 538083 & 819 & 538083           & 2.41     & 0.0118   & 3.42\\\hline  
3     & 2            & (3,1)     & 3 & 9       & 819      & 258  &  270          & 0.0118  & 0.0151  & 0.00663\\
4     & 2            & (6,1)     & 3 & 30     & 47988   & 27930 &  4752      & 0.923    &  2.16   & 0.0977\\
5     & 2            & (9,1)     & 3 & 51     & 538083  & 378504  & 24786  & 14.7     &   11.9  & 0.686\\
6     & 2            & (12, 1)  & 3 & 72     & 3006864 & 2317524 & 79488 & 184 &  88.5 &  2.79\\
7     & 2            & (15, 1)  & 3 & 93     & -& -& 195750 &  - & -& 8.38\\
8     & 2            & (18, 1)  & 3 & 114   & -&- & 408240 & - & - & 20.5\\\hline
9     & 2            & (2,2)     & 2 &  32   & 65792   & 65792 & 65792  & 0.284 & 0.280 &  0.799\\
10   & 2            & (2,2)     & 3 &  2080 & 16843008 & 16843008& 16843008 & 65.4 & 96.1 & 109\\
11   & 2            & (2,2)     & 4 & -    & - & - & -      & - & - & - \\\hline
12   & 2            & (2,2)     & 2 & 6       & 272 & 156 & 48 & 0.0216 & 0.00663 & 0.00420\\
13   & 2            & (3,3)     & 3 & 24     & 538083 &  378504  & 2754   &   15.9 & 18.4 & 0.121\\  
14   & 2            & (4,4)     & 4 & 60     & -          & -            & 298240& -         & -            & 25.0 \\\hline
15   & 2            & (4,2)     & 2 & 5      & 4160 & 812 & 480 &  0.203 & 0.0352  &  0.0206\\
16   & 2            & (4,2)     & 3 & 17     & 266304            & 22764 & 6624                 & 11.1          &     1.58   &  0.266\\
17   & 2            & (4,2)     & 4 & 21     & 17043520            & 637420 & 88544                 & 449          &  46.8        & 3.63 \\
18   & 2            & (4,2)     & 5 & 21    & -           & - & 1137120           & -       & -             & 48.2\\\hline
19   & 3            & (5,2)     & 3 &  260 & -           & - & 10850               & -       & -             & 1.75\\
20   & 4            & (5,2)     & 3 & 1122 & -          & - &  34650              & -       & -             & 6.66\\
21   & 5            & (5,2)     & 3 & 4952 & -          & - & 112850             & -        & -             & 25.5 \\
22   & 6            & (5,2)     & 3 & 22410 & -          & - & 376650           & -        & -             & 101.45
\end{tabular}
\label{tab:experimental-results}
\end{table}
We assess the scaling behavior by observing the trends in running time as parameters increase.
Since the optimized algorithms have higher cost per trace, they can be \emph{slower} when little speedup is achievable, e.g., in tests 1,2,9,10,11.
This is expected for 1,2 because they are not driving scenarios.
Likewise, the hazard scenario (9,10,11) is designed to provide no optimization opportunity and to use an atomic proposition.
By demonstrating the poor scalability of atomic propositions, it motivates the use of nominals.
In the other test scenarios, the baseline algorithm times out on larger input sizes, while the optimized and especially motion algorithms continue to perform better.
The motion algorithm performs particularly well for long durations, where its branching factor reductions relative to the optimized algorithm have the opportunity to compound, and in the platooning scenario where there is extensive fixed movement.

\section{Conclusions}
\label{sec:conclusion}
We developed Hybrid Spatiotemporal Logic (HSTL) for automotive safety, focused on highways.
We modeled scenarios like following, passing, intersection-crossing, and platooning.
We developed model-checking algorithms which exploit vehicle dynamics to minimize search space, proved them correct, and evaluated them on the scenarios.

Future work will explore using HSTL as a glue layer between planning and control.
Motion plans can be expressed as traces.
Likewise, a controller's resulting trajectory can be abstracted to a  grid and monitored for plan compliance.

\begin{credits}
\subsubsection{\ackname} 
Radu-Florin Tulcan is funded by the Vienna Science and Technology Fund (WWTF) project ICT22-023, Grant No. 10.47379/ICT22023. 
Yo\`av Montacute is supported by ACT-X, Grant No. JPMJAX24CR, JST, Japan.
Radu-Florin Tulcan, Yoàv Montacute, Kevin Zhou and Ichiro Hasuo are supported by ASPIRE, Grant No. JPMJAP2301, JST, Japan.
Yusuke Kawamoto is supported by FOREST, Grant No. JPMJFR242M, JST, Japan.
\subsubsection{\discintname}
The authors have no competing interests to declare that are
relevant to the content of this article. 
\end{credits}

\bibliographystyle{splncs04}
\bibliography{references}

\newpage 

\appendix

\section{Proofs}
\label{sec:proofs}
Note that the main paper and appendices may vary, for example, in their choice of variable names.
\subsection{Proofs for validities}
\label{sec:proof:validity}
In this section, we show the proofs for the validities presented in Section~\ref{subsec:valities}.

Horizontal and vertical moves commute on every grid point:
\begin{proposition}[Commutation of orthogonal moves]
\label{prop:valid1}
For any formula $\varphi$,
\begin{enumerate}
	\item $\models \opF\opR\varphi\leftrightarrow\opR\opF\varphi$
	\item $\models \opF\opL\varphi\leftrightarrow\opL\opF\varphi$
	\item $\models \opB\opR\varphi\leftrightarrow\opR\opB\varphi$
	\item $\models \opB\opL\varphi\leftrightarrow\opL\opB\varphi$.
\end{enumerate}
\end{proposition}
\begin{proof}
We show the proof for the first claim as follows.
Let $G$ be a grid-graph, $\vec{t}$ be a trace, and $p_{i,j} \in \Pos$.
To show the direction from left to right, assume that $G,\vec{t}, p_{i,j} \models \opF\opR \varphi$.
By the definition of the semantics, $p_{i+1,j}\in \Pos$ and $G,\vec{t}, p_{i+1,j} \models \opR \varphi$, hence $p_{i+1,j+1}\in \Pos$ and $G,\vec{t}, p_{i+1,j+1} \models \varphi$.
Then, $p_{i,j+1}\in \Pos$ and $G,\vec{t}, p_{i,j+1} \models \opF \varphi$.
Therefore, $G,\vec{t}, p_{i,j} \models \opR\opF \varphi$. The other direction can be derived analogously.

The other three claims can be proven in an analogous way.
\qed
\end{proof}

Traversing the four sides of a unit square and returning to the start yields the original truth value, whenever the entire cycle exists:
\begin{proposition}[Loops on the grid]
\label{prop:valid2}
For any formula $\varphi$,
\begin{enumerate}
	\item $\models \opF\opR\opB\opL\varphi\rightarrow\varphi$
	\item $\models \opR\opF\opL\opB\varphi\rightarrow\varphi$
	\item $\models \opF\opL\opB\opR\varphi\rightarrow\varphi$
	\item $\models \opB\opR\opF\opL\varphi\rightarrow\varphi$.
\end{enumerate}
\end{proposition}
\begin{proof}
We show the proof for the first claim as follows.
Let $G$ be a grid-graph, $\vec{t}$ be a trace, and $p_{i,j} \in \Pos$.
Assume that $G,\vec{t}, p_{i,j} \models \opF\opR\opB\opL\varphi$. 
By the definition of the semantics, $p_{i+1,j}\in \Pos$ and $G,\vec{t}, p_{i+1,j} \models \opR\opB\opL\varphi$, hence $p_{i+1,j+1}\in \Pos$ and $G,\vec{t}, p_{i+1,j+1} \models \opB\opL\varphi$.
Then $p_{i,j+1}\in \Pos$ and $G,\vec{t}, p_{i,j+1} \models \opL\varphi$.
Therefore, $p_{i,j}\in \Pos$ and $G,\vec{t}, p_{i,j} \models \varphi$.

The other three claims can be proven in an analogous way.
\qed
\end{proof}

Spatial modalities modify the grid coordinate while temporal modalities modify the trace index. 
We obtain the commutation of temporal with spatial moves, and the distribution of spatial moves along the until operator.
\begin{proposition}[Space and time interaction]
\label{prop:valid3}
For any formulas $\varphi$ and $\psi$,
\begin{enumerate}
	\item $\models \opX D\varphi\leftrightarrow D\opX\varphi$
	\item $\models \lozenge D\varphi\leftrightarrow D \lozenge \varphi$
	\item $\models \square D\varphi\leftrightarrow D \square \varphi$
	\item $\models D(\varphi\opU\psi)\leftrightarrow(D\varphi)\opU(D\psi)$.
\end{enumerate}
\end{proposition}
\begin{proof}
We show the proofs when $D$ is $\opF$. The proofs for the other spatial 
modalities are analogous.
Let $G$ be a grid-graph, $\vec{t}$ be a trace, and $p_{i,j} \in \Pos$.

We prove the first claim by the definition of the semantics as follows:
\begin{align*}
& G,\vec{t}, p_{i,j} \models \opX \opF\varphi
\\ \mbox{iff }~~
& \mbox{$\vec{t}^1$ is well-defined and  }
  G,\vec{t}^1, p_{i,j} \models \opF\varphi
\\ \mbox{iff }~~
& \mbox{$\vec{t}^1$ is well-defined, $p_{i+1,j}\in \Pos$, and }
  G,\vec{t}^1, p_{i+1,j} \models \varphi
\\ \mbox{iff }~~
& \mbox{$p_{i+1,j}\in \Pos$ and }
  G,\vec{t}, p_{i+1,j} \models \opX \varphi
\\ \mbox{iff }~~
& G,\vec{t}, p_{i,j} \models \opF\opX\varphi
\end{align*}

The second and the third claims can be proven analogously.

Next, we prove the last claim as follows.
By the definition of the semantics,
\begin{align*}
& G,\vec{t}, p_{i,j} \models \opF(\varphi\opU\psi)
\\ \mbox{iff }~~
& \mbox{$p_{i+1,j}\in \Pos$, and }
G,\vec{t}, p_{i+1,j} \models \varphi\opU\psi
\\ \mbox{iff }~~
& \mbox{$p_{i+1,j}\in \Pos$, and there is a $k$ s.t. $\vec{t}^k$ is well-defined, }
G,\vec{t}^k, p_{i+1,j} \models \psi,
\\
& \mbox{and }
G,\vec{t}^l, p_{i+1,j} \models \varphi  \text{ for all } \;  l <k
\\ \mbox{iff }~~
& \mbox{there is a $k$ s.t. $\vec{t}^k$ is well-defined, }
G,\vec{t}^k, p_{i,j} \models \opF \psi,
\\
& \mbox{and }
G,\vec{t}^l, p_{i,j} \models \opF \varphi  \text{ for all } \;  l <k
\\ \mbox{iff }~~
& G,\vec{t}, p_{i,j} \models (\opF\varphi)\opU(\opF\psi).
\end{align*}
\hfill\qed
\end{proof}

Several basic validities for standard hybrid logic~\cite{brauner2010hybrid} are also validities in our class of models:
\begin{proposition}[Validity with hybrid modality]
\label{prop:valid4}
For any formula $\varphi$ and any nominals $a, b, v\in\NP$,
\begin{enumerate}
	\item $\models \opBind{v} v$

	\item $\models \opAt{a} a$
	\item $\models \opAt{a} b \rightarrow \opAt{b} a$
	\item $\models \opAt{a} b \wedge \opAt{a} \varphi \rightarrow \opAt{b} \varphi$.
\end{enumerate}
\end{proposition}

\begin{proof}
Let $G$ be any grid-graph, $\vec{t} = ((x_0, y_0), \ldots, (x_n, y_n))$ be any trace, and $p_{i,j} \in \Pos$.
\begin{enumerate}
\item 
Notice that that $p_{i,j} = \subst{y_0}{z}{p_{i,j}}(z)$.
Then $G,\subst{\vec{t}}{z}{p_{i,j}}, p_{i,j} \models z$, and hence $G,\vec{t}, p_{i,j} \models \opBind{z} z$.
Therefore, $\models \opBind{z} z$.
\item 
By the definition of the semantics, $G,\vec{t}, y_0(a) \models a$, and hence $G,\vec{t}, p_{i,j} \models \opAt{a} a$.
Therefore, $\models \opAt{a} a$.
\item 
Assume that $G,\vec{t}, p_{i,j} \models \opAt{a} b$.
By definition, $G,\vec{t}, y_0(a) \models b$. Then $y_0(a)=y_0(b)$.
Hence $G, \vec{t}, y_0(b) \models a$, which then implies that $G,\vec{t}, p_{i,j} \models \opAt{b} a$.
Therefore, $\models \opAt{a} b \rightarrow \opAt{b} a$.
\item Assume that $G,\vec{t}, p_{i,j} \models \opAt{a} b \wedge \opAt{a} \varphi$.
Since $G, \vec{t}, p_{i,j} \models \opAt{a} b$, we obtain $y_0(a)=y_0(b)$.
Since $G, \vec{t}, p_{i,j} \models \opAt{a} \varphi$, we obtain $G, \vec{t}, y_0(a) \models \varphi$.
Substituting $y_0(b)$ for $y_0(a)$, we obtain that $G, \vec{t}, y_0(b) \models \varphi$, and hence $G,\vec{t}, p_{i,j} \models \opAt{b} \varphi$.
Therefore, $\models \opAt{a} b \wedge \opAt{a} \varphi \rightarrow \opAt{b} \varphi$.
\end{enumerate}
\qed
\end{proof}

The semantics of the hybrid modalities interacts nontrivially with spatial and temporal structure:
\begin{proposition}[Validity with binders]
\label{prop:valid5}
For any formulas $\varphi$ and $\psi$ and any nominal $v\in\NP$,
\begin{enumerate}
	\item $\models \opBind{v}\varphi\leftrightarrow\opBind{v}\opAt{v}\varphi$
	\item $\models \opBind{v}\opX\varphi\leftrightarrow\opX\opBind{v}\varphi$
	\item $\models \opBind{v}\lozenge\varphi\leftrightarrow\lozenge\opBind{v}\varphi$
	\item $\models \opBind{v}\square\varphi\leftrightarrow\square\opBind{v}\varphi$
	
	\item $\models \opBind{v}(\varphi\opU\psi)\leftrightarrow(\opBind{v}\varphi)\opU(\opBind{v}\psi)$
\end{enumerate}
\end{proposition}
\begin{proof}
Let $G$ be any grid-graph, $\vec{t}$ be any trace, and $p_{i,j} \in \Pos$.
\begin{enumerate}
\item 
By the definition of semantics, we have:
\begin{align*}
& G,\vec{t}, p_{i,j} \models \opBind{v}\varphi
\\ \mbox{iff }~~
& G,\subst{\vec{t}}{v}{p_{i,j}}, p_{i,j} \models \varphi
\\ \mbox{iff }~~
& G,\subst{\vec{t}}{v}{p_{i,j}}, \subst{y_0}{v}{p_{i,j}}(v) \models \varphi
\\ \mbox{iff }~~
& G,\subst{\vec{t}}{v}{p_{i,j}}, p_{i,j} \models \opAt{v} \varphi
\\ \mbox{iff }~~
& G,\vec{t}, p_{i,j} \models \opBind{v}\opAt{v}\varphi
\end{align*}
Therefore, $\models \opBind{v}\varphi\leftrightarrow\opBind{v}\opAt{v}\varphi$.
\item 
By the definition of the semantics, we have:
\begin{align*}
    & G, \vec{t}, p_{i,j}, \models \opBind{v} \opX \varphi \\
    \mbox{iff }~~ & G, \vec{t}[v \mapsto p_{i,j}], p_{i,j} \models \opX \varphi \\
    \mbox{iff }~~ & \vec{t}^1[v \mapsto p_{i,j}] \text{ is well-defined and } G, \vec{t}^1[v \mapsto p_{i,j}], p_{i,j} \models \varphi \\
    \mbox{iff }~~ & \vec{t}^1[v \mapsto p_{i,j}] \text{ is well-defined and } G, \vec{t}^1, p_{i,j} \models \opBind{v} \varphi \\
    \mbox{iff }~~ & G, \vec{t}, p_{i,j} \models \opX \opBind{v} \varphi
\end{align*}

Therefore, $\models \opBind{v}\opX\varphi\leftrightarrow\opX\opBind{v}\varphi$.

\item The proof for Claim 3 is analogous to that for Claim 2.
\item The proof for Claim 4 is analogous to that for Claim 2.
\item By the definition of semantics, we have:
\begin{align*}
    & G, \vec{t}, p_{i,j} \models \opBind{v} (\varphi \opU \psi) \\
    \mbox{iff }~~ & G, \vec{t}[v \mapsto p_{i,j}], p_{i,j} \models \varphi \opU \psi \\
    \mbox{iff }~~ & \exists k \in \{0, \ldots, n\} \text{ such that } \\
    & \quad \vec{t}^k[v \mapsto p_{i,j}] \text{ is well-defined}, \\
    & \quad G, \vec{t}^k[v \mapsto p_{i,j}], p_{i,j} \models \psi, \text{ and} \\
    & \quad \forall l < k, \ G, \vec{t}^l[v \mapsto p_{i,j}], p_{i,j} \models \varphi \\
    \mbox{iff }~~ & \exists k \in \{0, \ldots, n\} \text{ such that } \\
    & \quad \vec{t}^k[v \mapsto p_{i,j}] \text{ is well-defined}, \\
    & \quad G, \vec{t}^k, p_{i,j} \models \opBind{v} \psi, \text{ and} \\
    & \quad \forall l < k, \ G, \vec{t}^l, p_{i,j} \models \opBind{v} \varphi \\
    \mbox{iff }~~ & G, \vec{t}, p_{i,j} \models (\opBind{v} \psi) \opU (\opBind{v} \varphi)
\end{align*}
Therefore, $\models \opBind{v} (\varphi \opU \psi) \leftrightarrow (\opBind{v} \psi) \opU (\opBind{v} \varphi)$.
\end{enumerate}
\end{proof}

\subsection{Proofs for Non-validities}
\label{sec:proof:non-validity}
In this section, we show the proofs for the non-validities presented in Section~\ref{subsec:non-valities}.

In contrast to the validities above, the following formulas
fail on some model in the class of HSTL models. 
Each failure witnesses a different characteristic of the interaction between spatial, temporal and hybrid operations. 
These non-validities are particularly useful in separating models and for general expressivity arguments.

\begin{proposition}[Time-sensitivity of $\opAt{}$]
\label{prop:non-valid-at}
For any formula $\varphi$ and any nominal~$v$,
$$\not\models\opAt{v}\varphi\rightarrow\opX\opAt{v}\varphi.$$
\end{proposition}
\begin{proof}
Consider $q\in\AP$, a point $p_1\in\Pos$ such that $G,\vec{t}, p_1\models v\wedge q$ and $G,\vec{t}^1, p_1\models v\wedge\neg q$.
Then $G,\vec{t},p_1\models \opAt{v}q$ but $G,\vec{t},p_1\not \models \bigcirc\opAt{v}q$
\qed
\end{proof}

\begin{proposition}[Non-commutation of $\opAt{}$ and $\lozenge$]
\label{prop:non-commute-at-future}
For any formula $\varphi$ and any nominal~$v$,
$$\not\models\opAt{v}\lozenge\varphi\leftrightarrow\lozenge\opAt{v}\varphi.$$
\end{proposition}
\begin{proof}
Consider arbitrary points $p_1,p_2\in\Pos$, the formula $\varphi:=h\in \AP$, and a trace $\vec{t}=\big((x_1,y_1),(x_2,y_2)\big)$, where $ x_1(h)=x_2(h)=\{p_2\}$, $y_1(v)=p_1$ and $y_2(v)=p_2$. 
Then $G, \vec{t},p_1\models \lozenge\opAt{v}\varphi$, but $G,\vec{t},p_1\not\models \opAt{v}\lozenge\varphi$ since $h$ never holds at~$p_1$. 
\qed
\end{proof}

\begin{proposition}[Non-commutation of $D$]
\label{prop:non-commute-at-bind}
For any formula $\varphi$ and any nominal~$v$,
$$\not\models\opAt{v}D\varphi\leftrightarrow D\opAt{v}\varphi,\hspace{1em}\not\models\opBind{v}D\varphi\leftrightarrow D\opBind{v}\varphi.$$
\end{proposition}
\begin{proof}
For the first claim, consider $q\in \AP$ and a point $p_1\in\Pos$ such that $G,\vec{t},p_1\models v \wedge \neg q \wedge D q$.
Then $G,\vec{t},p_1\models \opAt{v} Dq$ but $G,\vec{t},p_1\not \models D\opAt{v} q$.

For the second claim, consider $\varphi:=v$, resulting in a contradiction. 
\qed
\end{proof}

\subsection{Proof of Theorem~\ref{thm:correctness and runtime of evaluator}}
\label{sec:proof:MC}
\begin{proof}
    First, we prove correctness. 

    Given a formula $\varphi$, construct the parse tree of $\varphi$, in which each node corresponds to an occurrence of a subformula of $\varphi$. 
    Then, given a grid-graph $G$, trace $\vec{t} = ((x_0, y_0), \dots, (x_n,y_n))$, and position $p_{i,j}$, for each node $\varphi'$ in the parse tree, there is an associated position $p_{\varphi'}$ and updated trace $\vec{t}_{\varphi'}$ given by following the behavior of the spatial and hybrid operators. 

    To be precise, define the positions $p_{\varphi'}$ and trace $\vec{t}_{\varphi'}$ traversing down the parse tree in the following way:
    \begin{itemize}
        \item At the root $\varphi$, $p_{\varphi} = p_{i,j}$ and $\vec{t}_{\varphi} = \vec{t}$.
        \item At a node $\varphi' = \mathop{D} \varphi_1$, the node $\varphi_1$ gets assigned position $p_{\varphi_1}$ obtained by moving $p_{\varphi'}$ in the direction $D$. 
        The trace $\vec{t}_{\varphi_1}$ is the same as $\vec{t}_{\varphi'}$. 
        \item At a node $\varphi' = \opBind{v} \varphi_1$, the node $\varphi_1$ gets assigned trace $\vec{t}_{\varphi_1} = \vec{t}_{\varphi'}[v \mapsto p_{\varphi'}]$. 
        The position $p_{\varphi_1}$ is the same as $p_{\varphi'}$.
        \item At any other node, its children inherit its position and trace.
    \end{itemize}
    
    Now, we prove a stronger claim, from which correctness follows, since \texttt{\textsc{Eval}} returns the value of $\texttt{memo}[(0, \varphi)]$. 
    \begin{claim}
        Fix a grid-graph $G$, a trace $\vec{t} = ((x_0, y_0), \dots, (x_n,y_n))$, a time step $k \in [0, n]$, and a position $p_{i,j}$.
        Fix a formula $\varphi$, and a subformula $\varphi'$ in the parse tree of $\varphi$.
        At the end of the execution of \texttt{\textsc{Eval\_memo}}$(G, \vec{t}, k, p_{i,j}, \varphi)$, if \texttt{memo}$[(k', \varphi')]$ is defined, then 
        \[
            G, (\vec{t}_{\varphi'})^{k'}, p_{\varphi'} \models \varphi' \qquad \iff \qquad \texttt{memo}[(k', \varphi')] = \texttt{True},
        \]
        Moreover, at the end of execution, \texttt{memo}$[(k, \varphi)]$ is defined.
    \end{claim} 
    
    The proof of the claim is by induction on the structure of $\varphi$ and the length of $\vec{t}^k$ (i.e., the value of $n-k$).

    First, notice that the execution of \texttt{\textsc{Eval\_memo}}$(G, \vec{t}, k, p_{i,j}, \varphi)$ will only ever update the value of \texttt{memo}$[(k, p_{i,j}, \varphi)]$, so it suffices to show that this value is correct; other values will be correct by the induction hypothesis.
    If \texttt{memo}$[(k, p_{i,j}, \varphi)]$ is not already defined, then line \ref{line:evaluator:save result to memo} is executed and stores the value of \texttt{result} to \texttt{memo}$[(k, p_{i,j}, \varphi)]$, and hence this value is defined by the end of execution.
    Also, the value of \texttt{memo}$[(k, p_{i,j}, \varphi)]$ is always returned, so it is enough to check the value stored to \texttt{result}.

    The base cases $\top$, $a \in \AP$, and $v \in \NP$ follow directly from the definition of the semantics. 

    The Boolean operations $\neg$ and $\land$, the next operator $\opX$, and the hybrid at operator $\opAt{v}$  follow in a straightforward manner from the induction hypothesis and the definition of the semantics. 
    For $\opX$, we point out that it reflects the semantics that $\opX \varphi_1$ is False when evaluated at the end of a trace.

    \bigskip

    The case $\varphi_1 \opU \varphi_2$ mirrors the typical expansion law for $\opU$ in LTL (see e.g. \cite{baier2008principles}); however, for completeness, we give the full argument.
    First, suppose that \texttt{result} gets value True. 
    Then either \texttt{\textsc{Eval\_memo}}$(G, \vec{t}, k, p_{i,j}, \varphi_2)$ is true, or $k < n$, and both \texttt{\textsc{Eval\_memo}}$(G, \vec{t}, k, p_{i,j}, \varphi_1)$ and \texttt{\textsc{Eval\_memo}}$(G, \vec{t}, k+1, p_{i,j}, \varphi_1 \opU \varphi_2)$ evaluate to True.
    In the case that \texttt{\textsc{Eval\_memo}}$(G, \vec{t}, k, p_{i,j}, \varphi_2)$ is true, then $G, \vec{t}^k, p_{i,j} \models \varphi_2$ by induction hypothesis, and thus $G, \vec{t}^k, p_{i,j} \models \varphi_1 \opU \varphi_2$.
    Otherwise,  $k < n$, and both \texttt{\textsc{Eval\_memo}}$(G, \vec{t}, k, p_{i,j}, \varphi_1)$ and \texttt{\textsc{Eval\_memo}}$(G, \vec{t}, k+1, p_{i,j}, \varphi_1 \opU \varphi_2)$ evaluate to True. 
    By induction hypothesis, we have that $G, \vec{t}^k, p_{i,j} \models \varphi_1$ and $G, \vec{t}^{k+1}, p_{i,j} \models \varphi_1 \opU \varphi_2$. 
    A standard analysis of the semantics of $\opU$ yields that $G, \vec{t}^k, p_{i,j} \models \varphi_1 \opU \varphi_2$.

    Conversely, suppose that \texttt{result} gets value False. 
    This means that \linebreak
    \texttt{\textsc{Eval\_memo}}$(G, \vec{t}, k, p_{i,j}, \varphi_2)$ evaluates to False, and if $k < n$, then at least one of \texttt{\textsc{Eval\_memo}}$(G, \vec{t}, k, p_{i,j}, \varphi_1)$ and \texttt{\textsc{Eval\_memo}}$(G, \vec{t}, k+1, p_{i,j}, \varphi_1 \opU \varphi_2)$ evaluates to False. 
    In the case that \texttt{\textsc{Eval\_memo}}$(G, \vec{t}, k, p_{i,j}, \varphi_1)$ is False, then by induction hypothesis, $G, \vec{t}^k, p_{i,j} \not\models \varphi_1$, and hence $G, \vec{t}^k, p_{i,j} \not\models \varphi_1 \opU \varphi_2$. 
    Otherwise, \texttt{\textsc{Eval\_memo}}$(G, \vec{t}, k+1, p_{i,j}, \varphi_1 \opU \varphi_2)$ is False, so by induction hypothesis, $G, \vec{t}^{k+1}, p_{i,j} \not\models \varphi_1 \opU \varphi_2$, and hence $G, \vec{t}^k, p_{i,j} \not\models \varphi_1 \opU \varphi_2$. 

    \bigskip

    For the spatial operators $D \varphi_1$, with $D \in \{\opF, \opB, \opL, \opR\}$, first note that the position $p_{\varphi_1}$ (as defined earlier) is simply $p_{i,j}$ moved one position in the direction $D$. 
    This is exactly the same position as $p'$ from line \ref{line:evaluator:spatial update}, and hence \texttt{result} is true if and only if \texttt{\textsc{Eval\_memo}}$(G, \vec{t}, k, p', \varphi_1)$ is true, which by induction hypothesis holds if and only if $G, \vec{t}^{k}, p' \models \varphi_1$. 
    This holds if and only if $G, \vec{t}^{k}, p_{i,j} \models D \varphi_1$, as desired.

    \bigskip

    For the bind operator $\opBind{v} \varphi_1$, first note that the trace associated with $\varphi_1$ is $\vec{t}_{\varphi_1} = \vec{t}[v \mapsto p_{i,j}]$, which is the same as the resulting value of $\vec{t}'$ after completing the for loop in line \ref{line:evaluator:bind reassignment complete}. 
    Thus \texttt{result} is true if and only if \texttt{\textsc{Eval\_memo}}$(G, \vec{t}[v \mapsto p_{i,j}], p_{i,j})$ is true, which by induction hypothesis holds if and only if $G, \vec{t}[v \mapsto p_{i,j}], p_{i,j} \models \varphi_1$.
    This holds if and only if $G, \vec{t}, p_{i,j} \models \opBind{v} \varphi_1$, as desired.

    \bigskip
    
    This completes the proof of correctness. 
    It remains to show the proof of the bound on the runtime. 
    Since \texttt{\textsc{Eval\_memo}} utilizes memoization, it is necessary to find the number of distinct subproblems and multiply by the complexity of solving each subproblem. 
    For the number of subproblems, notice that the keys to the memoization table are of the form $(k, \varphi')$, with $0 \leq k \leq |\vec{t}|$, and $\varphi'$ a subformula of $\varphi$. 
    Therefore, the number of subproblems is $O(|\vec{t}| \cdot |\varphi|)$. 

    As for the complexity of each subproblem, it is easy to check that in a single recursive call of \texttt{\textsc{Eval\_memo}}, all cases except for the bind operator $\opBind{v} \varphi_1$ only perform a constant number of operations. 
    In the case of the bind operator, \texttt{\textsc{Eval\_memo}} will make a copy of $\vec{t}$ and edit $|\vec{t}|-k$ of the copy's positions, so this requires $O(|\vec{t}|)$ steps. 

    Thus in total, \texttt{\textsc{Eval\_memo}} runs in $O(|\vec{t}|^2 \cdot |\varphi|)$ time.
\qed
\end{proof}

\section{Proof of Theorem~\ref{thm:motion algorithm number of traces generated}}
\begin{proof}
    First, since we assume that the assignments of atomic propositions are empty at all times, we will abuse notation slightly and write $t(v)$ to denote the position of nominal $v$ in state $t$ instead of $y(v)$.
    
    Now, note that a trace is yielded exactly when \texttt{\textsc{Extend\_trace}} is called (line \ref{line:motion alg:trace generated}).
    In particular, the trace that is generated is the fifth coordinate of the input to \texttt{\textsc{Extend\_trace}}.
    We prove that the following statements are true during the execution of \texttt{\textsc{Generate\_traces}} for all values $k=1, \ldots, n$, by induction on $k$:
    \begin{enumerate}
        \item Whenever $\texttt{\textsc{Extend\_trace}}(G, A, k, n, t, \vec{t})$ is called, $\vec{t}$ has length $k$, and $t$ is the last state of $\vec{t}$.
        \item The trace generated when $\texttt{\textsc{Extend\_trace}}(G, A, k, n, t, \vec{t})$ is called satisfies $A$.
        \item Every trace of length $k$ satisfying $A$ is generated.
        \item There are at most $O(s \cdot \left( \prod_{v \in \NP} b_v \right)^{k-1})$ such traces.
    \end{enumerate}
    The theorem follows from items 2-4. 

    \bigskip

    \textbf{Base case:} $k = 1$. 

    For item 1, the only invocations of $\texttt{\textsc{Extend\_trace}}(G, A, 1, n, t, \vec{t})$ occur when line \ref{line:motion alg:yield length 1 trace} is executed for each initial state, in which $\vec{t}$ has length 1, and the argument for $t$ is exactly the only element of $\vec{t}$.

    For item 2, when $\texttt{\textsc{Extend\_trace}}(G, A, 1, n, t, \vec{t})$ is called, the value of $\vec{t}$ is of the form \texttt{[initial\_state]}, where \texttt{inital\_state} is given by \linebreak
    $\texttt{generate\_initial\_states}(G, \A{rel}, \A{global})$, which outputs those states which satisfy $\A{rel}$ and $\A{global}$. 
    For traces of length 1, static car and fixed motion assumptions do not impose any restrictions, so it is enough to just check relative motion and global state assumptions. 
    This also proves item 3, since every such trace will be generated.

    Item 4 is satisfied since $O(s \cdot \left( \prod_{v \in \NP} b_v \right)^{1-1}) = O(s)$, and $s$ is precisely the number of valid initial states. 

    \bigskip

    \textbf{Step case:} 
    Suppose that all four items are true for $1 \leq k < n$. 
    We want to show they are true for $k+1$.

    For item 1, we see that \texttt{\textsc{Extend\_traces}} is called with current length input $k+1$ in line \ref{line:motion alg:generate next trace} when executing $\texttt{\textsc{Extend\_traces}}(G, A, k, n, t, \vec{t})$. 
    By the induction hypothesis, the trace $\vec{t}$ has length $k$, so the argument $\vec{t}.\texttt{append}(t')$ has length $k+1$, as desired, and also has $t'$ as its last state. 

    For item 2, we must check that in line \ref{line:motion alg:generate next trace}, the trace $\vec{t}.\texttt{append}(t')$ satisfies $A$. 
    The line \ref{line:motion alg:check global assumptions} ensures that $t'$ satisfies $\A{global}$.
    By the induction hypothesis, $\vec{t}$ also satisfies $\A{global}$, and therefore $\vec{t}.\texttt{append}(t')$ must as well.
    To analyze the motion assumptions, we see that the loop in line \ref{line:motion alg:nominal choice computation loop} assigns a set of choices $C(v)$ to all non-dependent nominals. 
    Once this is done, the loop in line \ref{line:motion alg:nominal assignments loop} checks over all possible assignments of nominals, chosen from these sets $C(v)$.
    Finally, line \ref{line:motion alg:assign dependent nominals} fixes the position of dependent nominals according to the positions of dependee nominals and $\A{rel}$. 
    $\A{static}$ is satisfied since a static nominal $v$ is given only one choice for its position (line \ref{line:motion alg:choices for static cars}), which is its location in the previous timestep $y(v)$. 
    $\A{rel}$ is satisfied since \texttt{complete\_trace} in line \ref{line:motion alg:assign dependent nominals} respects $\A{rel}$. 
    $\A{fixed}$ is satisfied since a fixed motion nominal $v$ is given exactly the (valid) choices from its motion set $M$ (line \ref{line:motion alg:choices for fixed motion cars}). 

    For item 3, let $\vec{t}' = \vec{t}.\texttt{append}(t')$ be a trace of length $k+1$ satisfying $A$.
    Let $t$ denote the last state of $\vec{t}$. 
    By the induction hypothesis, $\vec{t}$ must have been generated at some point during execution, and thus $\texttt{\textsc{Extend\_trace}}(G, A, k, n, t, \vec{t})$ must have been executed.
    Therefore, it suffices to check whether or not $t'$ is generated as a valid extension for which line \ref{line:motion alg:generate next trace} is executed. 

    Since $\vec{t}'$ satisfies $\A{static}$, for every static nominal $v$, $t'(v) = t(v)$.
    This is an allowed choice (in particular, the only allowed choice) in line \ref{line:motion alg:choices for static cars}.
    For fixed motion nominals, since $\vec{t'}$ satisfies $\A{fixed}$, the position $t'(v)$ must be within the set $\{\mathop{\vec{D}} t(v) \in \Pos \mid \vec{D} \in M\}$, which is exactly the allowed choices in line \ref{line:motion alg:choices for fixed motion cars}. 
    For dependee nominals, since $\vec{t'}$ satisfies $\A{rel}$, $t'(v)$ must be so that all dependents will fit into $G$; these choices are allowed in line \ref{line:motion alg:choices for dependees}.
    For dependent nominals, their positions are determined by $\A{rel}$ and the positions of their respective dependees---this is handled in line \ref{line:motion alg:assign dependent nominals}. 
    For all other nominals, their positions are allowed to be chosen arbitrarily (line \ref{line:motion alg:choices for free nominals}), and so the position $\vec{t'}(v)$ is included as an allowed choice. 
    Finally, $t'$ passes line \ref{line:motion alg:check global assumptions} since $\vec{t'}$ satisfies $\A{global}$. 
    Hence $\vec{t}'$ is generated.

    For item 4, we must count the number of traces for which line \ref{line:motion alg:generate next trace} is called. 
    For each generated trace of length $k$, this occurs as many times as possible values for $t'$ are generated. 
    By induction hypothesis, there are at most $O(s \cdot \left( \prod_{v \in \NP} b_v \right)^{k-1})$ traces of length $k$ generated, so it remains to count the number of values of $t'$ which are generated. 
    Notice that the number of such $t'$ is at most the size of $\prod_v C(v)$ as in line \ref{line:motion alg:nominal assignments loop}. 
    By definition, $|C(v)| \leq b_v$ (except on dependent nominals, where $C(v)$ is not defined, but in that case $b_v = 1$). 
    Thus the number of such $t'$ is at most $\prod_{v \in \NP} b_v$. 
    Hence the total number of traces of length $k+1$ generated is at most
    \[
        O\left(s \cdot \left( \prod_{v \in \NP} b_v \right)^{k-1} \right) \cdot \prod_{v \in \NP} b_v = O \left( s \cdot \left(\prod_{v \in \NP} b_v \right)^{k} \right),
    \]
    which is as desired for item 4. 
    \qed
\end{proof}

\section{Test Cases}
This appendix documents the test cases used in the evaluation section, as implemented in our Python code.
Each test case is structured as a Python function which passes an HSTL specification to the \verb|run_evaluator| helper function, which has the following arguments:
\begin{verbatim}
"""
Returns for a given formula all (trace, points) tuples where 
the formula holds.

:param run_id: the id of the run
:param propositions: the list of propositions used in the formula
:param nominals: the list of nominals used in the formula
:param assumptions: the list of formulas used as assumptions
:param conclusions: the list of formulas used as conclusions
:param grid_size: the grid size used for building the traces
:param trace_max_length: the maximal length of traces to consider
:param show_traces: whether the (trace, point) tuples should be 
    displayed in the console for debugging
:param evaluator_function: model checker evaluation function, 
    which can be any of the 3 checkers
"""
\end{verbatim}
Do not be misled by the names \verb|assumptions| and \verb|conclusions|. All these formulas are treated the same in a logical sense. 
It is only a style convention that \verb|conclusions| contains the final system safety properties such as collision-avoidance, while \verb|assumptions| contains conditions meant to configure the scenario, such as specifications of vehicle trajectories.
The six test case functions also have arguments. For example, \verb|test_index| is simply the corresponding row number in the evaluation table from the paper.
We present the code for each test function.

\begin{verbatim}
def left_right_test(test_index: int,evaluator_function: Callable):
    """
    Tests a spatial validity.

   :param test_index: test index
   :param evaluator_function: function of the checker for 
       evaluating the formula
   """
    run_evaluator(test_index, [], ['z'], [], 
       ["G(Left(Right(z)) <-> Right(Left(z)))"], (3, 3), 3, False,
       evaluator_function)
\end{verbatim}

\begin{verbatim}
def same_name_test(test_index: int, evaluator_function: Callable):
    """
    Tests a hybrid formula.

   :param test_index: test index
   :param evaluator_function: function of the checker for 
      evaluating the formula
   :return:
   """
    run_evaluator(test_index, [], ['z', 'z1'], [], ["G (@z z1)"],
      (3, 3), 3, False, evaluator_function)
\end{verbatim}

\begin{verbatim}
def one_lane_follow_test(test_index: int, duration: int,
  road_length: int, evaluator_function: Callable):
    """
    Test whether vehicle can safely follow another in the same 
    lane. A detailed description can be found in 
    README.md/Experiments/One Lane Follow.

    :param test_index: test index
    :param duration: maximal length of the traces
    :param road_length: length of the one-lane road
    :param evaluator_function: function of the checker for 
      evaluating the formula
    """
    run_evaluator(test_index, [], ['z0', 'z1'],
      [# SV is initially at the start of the lane
       "@z0 !(Back 1)",  
       # POV always moves forward or stays put
       "G (@z1 ↓z2 ((! X 1) | X @z1  (z2 | Back z2)))", 
       "G (@z0 ↓z2 ((! X 1) | X (@z0 ((!z1 & Back z2 ) " + 
       "| (z2 & Front z1) ))))"],
       # SV Always moves forward if safe, stays put if POV ahead
       ["G(!(@z0 z1))"], (road_length, 1), duration, False, 
         evaluator_function)
\end{verbatim}

\begin{verbatim}
def hazard_test(test_index: int, duration: int, 
   evaluator_function: Callable):
    """
    Tests whether vehicle can avoid static hazard in presence 
    of another vehicle.Figure 3 in paper
    
    :param test_index: test index
    :param duration: maximal length of traces
    :param evaluator_function: function of the checker for 
      evaluating the formula
    """
    width = 2
    length = 2 
    def fronts(i: int, p: str):
        if i == 0:
            return "({})".format(p)
        else:
            return "(Front {})".format(fronts(i-1, p))
    def bfront(p: str):
        each = ["(({})->({}))".format(fronts(i+1,"1"),
            fronts(i+1,p)) for i in range(0,length)]
        return "({})".format(reduce((lambda x, acc: x+"&"+acc),
          each))
    def dfront(p: str):
        each = [fronts(i+1,p) for i in range(0,length)]
        return "({})".format(reduce((lambda x, acc: x+"|"+acc),
          each))
    p1="(Right z1) & {}".format(dfront("G h"))
    p2="(@z0 ↓z2 X @z0 ((Back z2) & (G ! h)))"
    p3="(@z0 ↓z2 X @z0((Left z2) & {} & {}))"
         .format(dfront("z1"),bfront("G ! h"))
    full="@z0 (({}) & (({}) U ({})))".format(p1,p2,p3)
    run_evaluator(test_index, ["h"], ["z0", "z1"], [], [full], 
        (length,width), duration, False, evaluator_function)
\end{verbatim}

\begin{verbatim}
def safe_intersection_priority(test_index: int, duration: int, 
    grid_size: int, evaluator_function: Callable):
    """
    Test whether vehicle can go through an intersection safely.
    A detailed description can be found in 
    README.md/Experiments/Safe Intersection with Priority.

    :param test_index: test index
    :param duration: maximal length of the traces
    :param road_length: width and length of the road
    :param evaluator_function: function of the checker for
      evaluating the formula
    """
    run_evaluator(test_index, [], ['z0', 'z1'],
                  [# z1 starts somewhere on the left border
                  "@z1 !(Left 1)",  
                  # z0 starts somewhere on the bottom border
                   "@z0 !(Back 1)",  
                   # Moves left-to-right always
                   "G (@z1 ↓z2 ((! X 1)| X @z1 (Left z2)))",  
                   "G (@z0 ↓z2 ((! X 1)| X @z0 ((!z1 & Back z2)"+
                   " | (z2 & Front z1) )))"],
                  # Moves bottom-to-top but stop  to avoid others
                  ["G(!(@z0 z1))"], (grid_size, grid_size), 
                  duration, False, evaluator_function)
\end{verbatim}

\begin{verbatim}
def safe_passing(test_index: int, duration: int, 
   road_length: int, evaluator_function: Callable):
    """
    Tests maneuvers of vehicles: speed-up, swerve left and right
     A detailed description can be found 
     in README.md/Experiments/Safe Passing.

    :param test_index: test index
    :param duration: maximal length of the traces
    :param road_length: width and length of the road
    :param evaluator_function: function of the checker for 
       evaluating the formula
    """
    # z0 initially moves forward
    first_forward = "(@z0 ↓z2 ((! X 1) | X @z0 (Back z2)))"
    # then swerves left to avoid z1
    dodge_left = ("(@z0 ↓z2 ((Front z1) & "+
       "((! X 1)| X (@z0 (Back (Right z2))))))")
    # then drives twice as fast
    fast_forward = "(@z0 ↓z2 ((! X 1)| X @z0 (Back (Back z2))))"
    # then dodges back when safe
    dodge_right = "(@z0 ↓z2 ((! X 1)| X @z0 (Back (Left z2))))"
    # then drives normally
    last_forward = "(@z0 ↓z2 ((! X 1) | X @z0 (Back z2)))"
    run_evaluator(test_index, [], ['z0', 'z1'],
      [# POV starts anywhere in right lane, stays in right lane
       "G(@z1 !(Right 1))",  
       # SV starts in back of right lane
       "@z0 !(Right 1)",  
       "@z0 !(Back 1)",
       # z1 moves forward or stays in place
       "G (@z1 ↓z2 ((! X 1) | X @z1  (z2 | Back z2)))",  
       "({} U ({} & ((! X 1) | X ({} & ((! X 1) | X ({} U ({} & "+
       "((! X 1) | X G ({})))))))))".format(
         first_forward, dodge_left, fast_forward, fast_forward, 
         dodge_right, last_forward)],
      ["G(!(@z0 z1))"], (road_length, 2), duration, False,
         evaluator_function)
\end{verbatim}

\begin{verbatim}
def join_platoon(test_index: int, duration: int, 
   platoon_size: int, road_length: int, 
   evaluator_function: Callable):
    """
    Tests safe joining of a vehicle to a platoon of vehicles.

    :param test_index: test index
    :param duration: maximal length of the traces
    :param platoon_size: size of vehicle platoon
    :param road_length: width and length of the road
    :param evaluator_function: function of the checker for 
       evaluating the formula
    """
    pov_noms = ["z" + str(i + 1) for i in range(platoon_size)]
    noms = ["z0"] + pov_noms  # and z is a temporary
    no_collide = "!({})".format(reduce((lambda x, acc: x+"|"+acc),
      pov_noms))
    each_front = ["Front " + n for n in pov_noms]
    some_front = reduce((lambda x, acc: x + "|" + acc), 
      each_front)
    sv_mov_assump = (
       "G(@z0 ↓z ((! X 1) | (X @z0((Back z)|(({0})&"+
       "(Right z)&({1}))))))".format(some_front, no_collide))
    sv_start_assump = "@z0 !(Right 1)"
    pov_start_assumps = ([format("G(@z{0} !(Left 1))".
       format(str(i + 1))) for i in range(platoon_size)])
    pov_mov_assumps = (
      [format("G(@z{0}↓z ((! X 1)|X(@z{0} (Back z))))"
      .format(str(i + 1))) for i in range(platoon_size)])
    assumps = ([sv_start_assump, sv_mov_assump] + 
      pov_mov_assumps + pov_start_assumps)
    postcond = "G(@z0 ({}))".format(no_collide)
    run_evaluator(test_index, [], noms, assumps, [postcond],
     (road_length, 2), duration, False, evaluator_function)
\end{verbatim}

The evaluation table was prepared by running the following main function and  extracting the output of \verb|run_evaluator|. In the code, the motion algorithm is \verb|evaluate_optimized2|.

\begin{verbatim}
if __name__ == '__main__':
    for funct in [evaluate_baseline, 
                  evaluate_optimized1,
                  evaluate_optimized2]: 
        # Test 1
        left_right_test(1, funct)
        # Test 2
        same_name_test(2, funct)
        # Test 3
        one_lane_follow_test(3, 3, 3, funct)
        one_lane_follow_test(4, 3, 6, funct)
        one_lane_follow_test(5, 3, 9, funct)
        one_lane_follow_test(6, 3, 12, funct)
        one_lane_follow_test(7, 3, 15, funct)
        one_lane_follow_test(8, 3, 18, funct)
        # Test 4
        hazard_test(9, 2, funct)
        hazard_test(10, 3, funct)
        # Test 5
        safe_intersection_priority(11, 2, 2, funct)
        safe_intersection_priority(12, 3, 3, funct)
        safe_intersection_priority(13, 4, 4, funct)
        # Test 6
        safe_passing(14, 2, 4, funct)
        safe_passing(15, 3, 4, funct)
        safe_passing(16, 4, 4, funct)
        safe_passing(17, 5, 4, funct)
        # Test 7
        join_platoon(18, 3, 2, 5, funct)
        join_platoon(19, 3, 3, 5, funct)
        join_platoon(20, 3, 4, 5, funct)
        join_platoon(21, 3, 5, 5, funct)
\end{verbatim}
\end{document}